\documentclass[12pt]{article}
\usepackage{titling}
\newcommand{\cO}[1]{}
\PassOptionsToPackage{hyphens}{url}
\usepackage{csquotes}
\usepackage[
    authordate,
    backend=biber,
    natbib=true
]{biblatex-chicago}

\usepackage[doublespacing]{setspace}
\usepackage{hanging}
\usepackage{amsmath, amssymb}
\usepackage{amsthm}
\usepackage{bbm}
\usepackage[hyphens]{xurl}
\usepackage{hyperref}
\hypersetup{
    colorlinks,
    linkcolor={red!50!black},
    citecolor={orange!50!black},
    urlcolor={blue!80!black}
}
\newcounter{assum}
\newtheorem{prop}{Proposition}
\newtheorem{assumption}[assum]{Assumption}

\usepackage{empheq}

\newtheorem{remark}{Remark}

\usepackage{appendix}
\usepackage{units}
\usepackage{enumerate}
\usepackage{booktabs}
\usepackage{tabularx}
\usepackage{arydshln}
\usepackage{amssymb}
\usepackage[dvipsnames]{xcolor}
\usepackage{graphicx}
\usepackage{subfigure}
\usepackage{tikz}
\usepackage{enumitem}
\usepackage{pgfplots}
    \usepackage{caption}
\usetikzlibrary{decorations.pathreplacing,calligraphy}
\pgfplotsset{width=10cm,compat=1.9}

\newcommand{\argmin}{\mathop{\rm arg~min}\limits}

\usepackage{accents}

\usepackage[top=25truemm,bottom=25truemm,left=25truemm,right=25truemm]{geometry}

\definecolor{green}{HTML}{30AE17}

\begin{document}

\title{Designing Signals for Deterrence}

\author{
\begin{tabular}{cc}
Hiroto Sawada &\quad\quad\quad\quad\quad Raphael Boleslavsky \\[0.3em]
Kristopher W. Ramsay &\quad\quad\quad\  Mehdi Shadmehr\thanks{
Sawada: Department of Politics, Princeton University 
(e-mail: \texttt{hsawada@princeton.edu});
Boleslavsky: Department of Business Economics \& Public Policy, 
Kelley School of Business, Indiana University 
(e-mail: \texttt{rabole@iu.edu});
Ramsay: Department of Politics, Princeton University 
(e-mail: \texttt{kramsay@princeton.edu});
Shadmehr: Department of Public Policy, University of North Carolina at Chapel Hill 
(e-mail: \texttt{mshadmeh@gmail.com}).
}
\end{tabular}
}

\date{}
\maketitle

\begin{abstract} 
\begin{singlespace}
States and alliances spend enormous resources signaling resolve to deter adversaries. Two canonical strategies are burning money (sunk costs) and burning bridges (audience costs). Large audience costs can deter effectively, but they may be infeasible and, even when feasible, can increase the likelihood of war relative to burning money. We propose a signaling design in which the defender retains the flexibility to fight or back down but commits ex ante, before learning its resolve, to a rule mapping realized resolve into costly actions. When audience costs are limited, this design raises defender payoffs while reducing the likelihood of war. We discuss implementation through laws and automated algorithms. We also characterize defender payoffs when audience costs are chosen endogenously from a feasible set, showing that lowering the maximum feasible audience cost can raise defender payoffs, and we compare these payoffs with those under the standard burning-money setting without commitment.

\end{singlespace}

\vspace{.3in}


\vspace{.1in}


\end{abstract}

\let\markeverypar\everypar
 \newtoks\everypar
 \everypar\markeverypar
 \markeverypar{\the\everypar\looseness=-2\relax}
\thispagestyle{empty}
\parskip=2.8pt

\thispagestyle{empty}
\newpage
\baselineskip=20.pt

\setcounter{page}{1}

\section{Introduction}

Following Russia's seizure of Crimea, Deputy Secretary of Defense Robert Work explained that the European Reassurance Initiative---which funded additional rotational deployments and infrastructure in Eastern Europe---was intended ``to signal our resolve and capacity to deliver'' \citep{workCSISGlobalSecurity2014}. Two years later, Secretary of Defense Ash Carter similarly stated, ``We're also increasing military exercises with allies and partners to demonstrate resolve and build their resilience,'' while also emphasizing their contribution to interoperability \citep{carterStrategicTransition2016}.
Similarly, as NATO prepared in 2016 to establish an enhanced Forward Presence (eFP) in the eastern part of the Alliance, Secretary General Jens Stoltenberg explained that it ``shows our resolve. It enhances our deterrence'' \citep{stoltenbergWarsawSummit2016}. He likewise described the purpose of NATO's Trident Juncture 2018 exercise not only as ``to strengthen our ability to operate together,'' but also as ``to send a clear message. NATO has the capabilities and the resolve to protect all Allies against any threat'' \citep{stoltenbergTridentJuncture2018}. Following Russia's invasion of Ukraine in 2022, NATO placed about 40,000 troops under NATO command in the eastern part of the Alliance and more than 100,000 troops on heightened alert. Stoltenberg explained: ``This is to send a very clear message to Moscow and leave no room for miscalculation or misunderstanding: an attack on one NATO Ally will trigger the full response from the whole Alliance. This is deterrence'' \citep{stoltenbergSecuringEurope2022}.

Indeed, this logic is embedded in NATO doctrine. NATO's \emph{Allied Joint Doctrine} identifies credibility and communication among the five principles of deterrence \citep[p.~43]{natoAlliedJointDoctrine2022}. It states that NATO must be perceived as possessing both the ``collective will to act'' and the capability to do so, and that, in some circumstances, the joint force should take considered risks to ``communicate the strength of the Alliance's commitment'' \citep[p.~44]{natoAlliedJointDoctrine2022}; see also \citet[p.~19]{janzenAlliedCommandOperations2022}. NATO's current deterrence guidance similarly states that exercises contribute to deterrence by communicating the Alliance's ``capability, readiness and resolve to potential adversaries'' \citep{natoDeterrenceDefence2026}.

As these examples highlight, to deter an adversary, a state must make its willingness to defend sufficiently credible. But assertions may not be credible: a state that would ultimately back down has the same incentive to say that it will fight. States can make such claims more credible in two fundamentally different ways. They can take costly actions before a challenge occurs—mobilizing forces, conducting military exercises, accumulating weapons, or undertaking other costly preparations. Or they can take actions that make backing down costly if a challenge later occurs---for example, by publicly staking their reputation on a threat.\footnote{States can also increase their strength. Even absent any information transmission incentives, states take costly actions to raise their preparedness. For example, if the state can virtually ensure that it will win a fight against an adversary, that may guarantee deterrence. Our focus is on what states can do conditional on a preparedness level. More generally, the marginal gains from preparedness actions eventually fall below their marginal costs. When information transmission incentives are present, states then may go beyond those levels to communicate information even though such costly actions have minimal effect on preparedness.} We call the first strategy \emph{burning money} and the second \emph{burning bridges} or \emph{creating audience costs}. \citet{fearonSignalingForeignPolicy1997} refers to the corresponding signaling mechanisms as ``sinking costs'' and ``tying hands.''

Which is the better way to deter an adversary: burning money or burning bridges (creating audience costs)? Burning bridges has an apparent advantage. Money burned is lost regardless of what happens next. A cost created by burning bridges, by contrast, is incurred only if the state is challenged and subsequently backs down. If the threat succeeds in deterring the challenger, or if the state fights when challenged, the cost will not be paid. When creating audience costs is feasible, it is therefore a remarkably effective approach.

This logic underlies \citeauthor{fearonSignalingForeignPolicy1997}'s \citeyearpar{fearonSignalingForeignPolicy1997} result that tying hands yields a higher expected payoff on average than sinking costs. An empirical literature finds evidence consistent with audience costs \citep{tomzDomesticAudienceCosts2007,trager_PoliticalCostsCrisis2011,levyBackingOutBacking2015,liPublicOpinionInternational2021,gelpiWinnersLosersDemocracies2001,weeksAutocraticAudienceCosts2008,takeiAudienceCostsCredibility2024,crismanCoxAudienceCostsDynamics2018}. However, evidence that audience costs are typically large enough to generate commitment effects is thinner \citep{snyderCostEmptyThreats2011,katagiriCredibilityPublicPrivate2019,trachtenbergAudienceCostsHistorical2012,gartzkeStillLookingAudience2012,crocoWhatCostReexamining2021}; see \citet[pp.~46-62]{slantchev2011book} and \citet{hoshiroThirtyYearsAudience2026} for reviews and \cite{tangEmpiricalEvidenceAudience2026} for a meta-analysis. Moreover, when some challengers are undeterrable, even sufficiently large audience costs can increase the probability of war relative to standard sunk-cost signaling by locking low-resolve defenders into fighting when they would otherwise concede. This raises a natural question: are there better ways for states to signal resolve?

We introduce a new form of commitment into the classical burning-money approach that simultaneously increases the state's payoff and (weakly) reduces the likelihood of war. What if a state can commit, \emph{before it learns its resolve}, to how it will burn money once its resolve is realized? That is, before knowing how much it will value the issue at stake in a future confrontation, the state chooses a contingent signaling strategy that maps its subsequently realized resolve into observable costly actions.\footnote{This approach applies \citet{BoleslavskyShadmehr_Signaling_with_Commitmentpdf}'s game theoretical framework to crisis-signaling models in international relations.} The state does not commit to fight. It retains the choice between fighting and backing down if a challenge actually occurs. Instead, it commits to \emph{how it will signal its willingness to fight} after learning its resolve but before the adversary decides whether to challenge. The next section discusses how laws and automated algorithms can create this commitment power.

This form of commitment is conceptually distinct from burning bridges. Burning bridges restricts the state's future choices by changing the consequences of backing down. Commitment to a burning-money strategy instead restricts the state's future signaling choice: it determines how the state will communicate after learning information but before the adversary acts. A state can therefore commit to its eventual response, or it can commit to the process by which its information is revealed. We characterize when each form is more valuable.

This form of commitment changes the logic of burning money. Without commitment, a costly signal must be sufficiently expensive to discourage imitation. A resolved state burns money to distinguish itself from an unresolved state, and the amount it burns must be large enough that the unresolved state prefers not to mimic it. The cost of the action is therefore precisely what makes the action informative.

With ex ante commitment, the costly action no longer has to perform this function. The state chooses the mapping from future resolve into actions before knowing which resolve will be realized. It can therefore assign different observable actions to its future types in advance. Conditional on being able to commit to that assignment, the actions need only be distinguishable; their difference in cost need not be large enough to deter imitation ex post. Therefore, the state uses the smallest available positive costly action. Commitment strictly improves the state's ex ante payoff because it preserves the informational benefit of separation while reducing the resources that must be burned.

This logic generates a striking reversal of the usual interpretation of costly signals. When the state is relatively unlikely to be highly resolved, the familiar pattern obtains: the high-resolve state takes the costly action and the low-resolve state does not. But when high resolve is sufficiently likely, the optimal assignment reverses. The high-resolve state takes the costless action and the relatively unlikely low-resolve state burns money. The absence of a costly action then signals \emph{greater} resolve. We refer to this feature as \textit{reversal of meaning}. This reversal of meaning is possible because the meaning of an action is determined by the ex ante commitment to the signaling rule, rather than by a type's ex post willingness to incur its cost.

 Does this strengthened form of burning money overturn the advantage of burning bridges (creating audience costs)? The answer depends on how large a bridge the state is capable of burning---how large an audience cost the state can create. 
We show that sufficiently large audience costs retain a powerful advantage over burning money. But that advantage can arise through two distinct mechanisms. In the first and most familiar mechanism, burning bridges does not reveal resolve; it creates ex post resolve. An audience cost can be large enough that even the low-resolve defender, who would otherwise back down, prefers to fight rather than incur the cost of retreat. Both types can then credibly stand firm. When all challengers are deterrable, this approach trivially creates the highest payoff for the defender. However, even when undeterrable challengers may arise, low-resolve types create no externality for high-resolve ones:  once they create large enough audience costs to lock themselves into fighting, they will be treated as high-resolve by challengers.

In the second mechanism,
burning bridges works as a signaling device.  A sufficiently large prospective cost of backing down makes the low-resolve defender unwilling to imitate the high-resolve defender. The two types separate, but the audience cost is never actually paid on the equilibrium path: the high-resolve type fights if challenged, while the low-resolve type does not send the signal. Burning bridges therefore produces the same information as burning money without dissipating resources along the equilibrium path.

Whenever the defender can generate an audience cost large enough to achieve at least one of these two objectives---separating the types or committing the low-resolve type to fight---burning bridges yields a strictly higher ex ante payoff than burning money. This is the best the defender can hope for under standard conditions.\footnote{In particular, this holds when the defender's interim payoffs are convex, as we will show is the case with standard distributions of the challenger's valuation of the contested object when the expected value is large compared to the costs of conflict.}
Therefore, this ranking survives even with the proposed new form of commitment. Ex ante commitment makes burning money cheaper, but it cannot reproduce these advantages when sufficiently large audience costs are feasible.

Importantly, this ranking reverses when sufficiently large audience costs are infeasible. With low audience costs, the low-resolve defender now has an incentive to mimic the high-resolve defender and the equilibrium involves pooling or semi-separation: the low-resolve defender may bluff by taking the same action as the high-resolve defender and then back down when challenged.
Consequently, when audience costs are too low to prevent bluffing (i.e., either to induce separation or to change subsequent behavior), committing to how money will be burned provides the defender with a strictly higher ex ante payoff for every interior prior. The reason is not that the state can burn more money. In fact, commitment makes large expenditures unnecessary. By assigning actions to future types before learning its resolve, the state can use the smallest available costly action to communicate information. Even without commitment, burning money outperforms constrained audience costs if and only if the low-resolve conditions, and hence bluffing and backing down, arise with sufficient frequency.

The two signaling approaches also have different implications for war. Commitment to a burning-money strategy changes the cost---and sometimes even the meaning---of the signal, but it does not change the equilibrium probability of war. With and without commitment, burning money ultimately separates high and low resolve, so war occurs only when a high-resolve defender encounters an undeterrable challenger. 

The audience-cost approach leads to war at least as often as burning money, but for different reasons. When a large audience cost separates the types without committing the low-resolve defender (second mechanism above), burning bridges and burning money generate the same probability of war. When the audience cost instead works by committing the low-resolve defender to fight (first mechanism), burning bridges generates strictly more war: some challenges that would have ended in acquiescence now lead to fighting. This is the familiar danger of tying hands emphasized by \citet{fearonSignalingForeignPolicy1997}---the same commitment that strengthens deterrence can leave a state locked into fighting when deterrence fails.

When audience costs are too small to sustain separation, they also generate strictly more war, but for the opposite reason. They do \emph{not} successfully tie the low-resolve defender's hands. Instead, they allow bluffing. Because a signal can come from either type, the challenger sometimes tests a defender it believes may be unresolved and instead encounters the high-resolve type, which fights. Thus, excess war under burning bridges can arise from two very different failures of discretion: sufficiently strong audience costs can commit too many types to fight, while insufficiently strong audience costs can leave too much uncertainty about which type the challenger faces.

This analysis points to a different way of thinking about the resources that support credible signaling. One might expect the choice between burning money and burning bridges (creating audience costs) to depend on both political institutions that create commitment powers (e.g., audience costs) and economic capacity that enables richer states to sink larger costs. However, once a sufficiently rich set of burning-money actions includes a small positive action, the ability to burn still larger amounts is not what makes this approach effective. 
With commitment, the state wants to use the cheapest distinguishable signal. Therefore, it is the defender's ability to commit to a burning-money strategy, largely affected by political institutions, not the size of the money burned, that controls the efficacy of burning money for the defender. By contrast, the largest bridge the state can burn (the maximum feasible audience cost), also largely affected by political institutions, is decisive. Whether audience costs are large enough to sustain separation or alter subsequent incentives determines both the payoff ranking and, in part, the probability of war. The comparison, therefore, shifts attention from the sheer capacity to dissipate resources toward the institutional capacity.

Our main contribution is to introduce a new form of commitment to burning-money signaling protocols that is robustly better than signaling through audience costs in settings where states cannot create sufficiently large audience costs to eliminate bluffing—full commitment and full separation through audience costs are both infeasible. This advantage is robust in the sense that it applies to any (interior) prior and in the sense that it features a lower probability of war. We show how such a protocol leads to the reversal of the meaning of money burning, where low-resolve types burn money while high types do not. Our results apply the techniques developed in \citet{BoleslavskyShadmehr_Signaling_with_Commitmentpdf} to crisis signaling settings—see \citet{horzCommunityInterventions2025},  \citet{chenSociallyOptimalExante2025}, and \citet{boleslavskyAlgorithmTransparency2026} for applications to law and online sales platforms.  We discuss the feasibility of such commitment protocols in the next section.

Our secondary contribution is to characterize equilibrium payoffs in settings with limited audience costs and compare them with payoffs in burning-money settings with and without commitment. A large theoretical literature investigates various aspects of crisis signaling, including the endogenous determination of audience costs \citep{smithInternationalCrisesDomestic1998,acharyaBehavioralFoundationAudience2019,ashworthAccountabilityPoliticiansInternational2024} and  variants of the canonical crisis signaling model with limited audience costs  \citep{tararLimitedAudienceCosts2013,kurizakiDetectingAudienceCosts2015}. For example, \citet[][p.~954]{kurizakiDetectingAudienceCosts2015} consider an exogenously fixed audience cost and show that as the audience cost shrinks, the equilibrium moves from separating to semi-separating to pooling. Our contribution here is to show that the defender’s payoff is non-monotone in the maximum feasible audience cost and that, when feasible audience costs are limited, the defender’s payoff under burning money without commitment exceeds that under audience costs if and only if the defender is a priori sufficiently likely to be low-resolve.

\section{Commitment Mechanisms for Signaling Resolve}
  
Our analysis highlights the value of constructing or strengthening mechanisms that enable commitment to mapping realized resolve into costly actions. In this section, we highlight two such mechanisms: laws and automated algorithms.

\paragraph{Laws} The government can pass laws that map the state's resolve (e.g., its valuation of the contested object or other parameters that determine its willingness to fight) into an action.
States already employ systems that combine preestablished plans with observable and costly actions that aim, at least partly, to signal resolve to defend against attacks. 
 For example, the U.S. DEFCON (defense readiness condition) system maps a level of threat to prespecified actions that are both costly and informative about U.S. intentions and resolve. Discussing DEFCON levels in a meeting of the Executive Committee of the National Security Council during the Cuban Missile Crisis, for example, Secretary of Defense Robert McNamara emphasized both the informational content and the material cost of U.S. readiness:\footnote{\href{https://history.state.gov/historicaldocuments/frus1961-63v11/d170}{\emph{Foreign Relations of the United States}, Volume XI, Document 170.}}
\begin{quote}
any reduction in the state of readiness of U.S. forces would be a sign to the Soviet Union. The Russians would know immediately if our state of alert was reduced. He suggested that no reduction be made today. However, SAC [Strategic Air Command] should stand down as soon as possible because the present alert involves burning out large amounts of spare parts.
\end{quote}

As another example, the federal regulation establishing Graduated Mobilization Response (GMR) states that ``[t]he GMR system enables the nation to approach mobilization planning and actions as part of the deterrent response capability and to use it to reduce the probability of conflict.''\footnote{\href{https://www.govinfo.gov/content/pkg/CFR-2025-title44-vol1/pdf/CFR-2025-title44-vol1-part334.pdf}{44 C.F.R.\ 334.3(c)}.} The 1990 \emph{National Security Strategy of the United States} similarly called for plans to ``develop graduated responses that will themselves signal U.S.\ resolve and thus contribute to deterrence.''\footnote{\href{https://history.defense.gov/Portals/70/Documents/nss/nss1990.pdf}{The White House, March 1990. \emph{National Security Strategy of the United States}.} Washington, DC: The White House, p.~27. See also \href{https://www.nationalacademies.org/read/28201/chapter/1}{National Research Council. 1991. \emph{Graduated Mobilization Response: Forging a Strong Partnership with Industry}.} Washington, DC: The National Academies Press.} Other mappings include DoD Force Protection Conditions (FPCON), NATO Alert System during the Cold War, and Department of Energy Graded Protection.

These examples illustrate relevant institutional building blocks---preplanned contingent responses, differentiated observable actions, and costly measures explicitly understood to convey information to an adversary. Our results push further in this direction. A state may benefit from committing---e.g., through legislation---to the mapping between subsequently realized conditions and costly signals, thereby surrendering some discretion over \emph{how it communicates} while retaining discretion over whether ultimately to fight.\footnote{In the burning-bridges approach, it is difficult to imagine how a country can commit to create audience costs, e.g., by requiring the government to make certain public statements, and no others, under certain conditions.}

\paragraph{Automated Algorithms}   The government can set up automated systems that evaluate resolve, map it into costly actions, and execute the actions. Automated systems have been in place since the Cold War and are analyzed particularly in the context of nuclear warfare. \cite{schwartzDelegatingDestruction2026} discuss specific programs and provide evidence that automated systems are perceived as a commitment device. In the burning money setting, automated algorithms could, e.g., issue warnings that induce costly actions, or even expend ammunition and equipment without causing casualties.\footnote{In other domains such as online sales platforms, computer algorithms are viewed as commitment mechanisms on the part of the seller \citep{boleslavskyAlgorithmTransparency2026}.}

\section{Model}

We begin by outlining the model informally before turning to its formal exposition. We consider deterrence settings with two canonical forms: burning money and audience costs (or burning bridges). The two settings have the same structure and differ only in their payoffs. In both cases, a challenger seeks to obtain an object currently held by a defender. If the challenger attacks and the defender does not defend, the object passes from the defender to the challenger. If the defender instead chooses to defend, war follows. War is costly to both sides, its outcome is uncertain, and the winner obtains the contested object.

The defender's valuation of the object is private information. Some defender types value the object enough to fight for it, while others do not. The defender's valuation can therefore be interpreted as its resolve to fight in order to preserve the status quo.

In the burning-money setting, the defender takes an early costly action that may signal how much it values the object. In the audience-cost setting, the defender instead takes an early action that becomes costly only if the challenger later attacks and the defender backs down rather than fighting.

In the burning-money setting, the defender can commit in advance to an action that depends on its valuation of the object. For example, when the rule of law is strong in a country, it can pass laws requiring certain costly actions upon learning the value of the contested object. The executive agencies responsible for carrying out those actions then act accordingly. By contrast, it is harder to interpret such a commitment power in audience-cost settings.

Formally, there is a defender and a challenger. First, the defender chooses an action $s\in S=\{0,s_1,s_2,\cdots,s_N\}$ with $0<s_1<s_2<\cdots<\overline s:=s_N$.  Next, the challenger observes the defender's action and then chooses $a\in\{0,1\}$, where $1$ means ``attack" and $0$ means ``don't attack." If $a=0$, the game ends. If $a=1$, then the defender chooses $b\in\{0,1\}$, where $1$ means ``fight" and $0$ means ``don't fight,'' payoffs are realized, and the game ends.

When we need to distinguish between the settings, we use the superscript $b$ for the burning-money setting and $t$ for the audience-cost setting. For example, $S^b$ refers to the set of available actions (signals) and $s^b_1$ to the smallest strictly positive action (signal) in the burning money setting, and $\overline s^t:=s^t_N$ refers to the largest action (signal) in the audience cost setting.

 There is a state of the world $(\omega,v_c)\in\{\omega^H,\omega^L\}\times \mathbb{R}_+$ with $0<\omega^L<\omega^H$, where $\omega, v_c>0$ are the values of the contested prize to the defender and challenger, respectively. We refer to the defender with $\omega=\omega^H$ as the high type and the one with $\omega^L$ as the low type.   Let $u(a,b,v_c)$  be the challenger's payoff, $v^b(s,a,b,\omega)$ be the defender's payoff in the burning money setting, and $v^t(s,a,b,\omega)$  be the defender's payoff in the audience-cost setting:
\begin{eqnarray}
u(a,b,v_c)&=&a\left(b\ ((1-p)v_c-c_c)+(1-b) v_c \right)\label{eqn:u}\\
v^b(s,a,b,\omega)&=& (1-a) \omega+ ab (p \omega-c_d)-s\label{eqn:v-b}\\
v^t(s,a,b,\omega)&=&(1-a) \omega+ a b (p \omega-c_d)-a(1-b)s,\label{eqn:v-t}
\end{eqnarray}
where $p\in(0,1)$ is the probability that the defender wins a war with the challenger, and $c_d,c_c>0$ are the costs of war to the defender and challenger, respectively.

The first component of the state, $\omega$, is the private information of the defender, and its second component, $v_c$, is the private information of the challenger. 
$\omega$ and $v_c$ are independent from each other, and players share a prior that $v_c\sim F$ and that $\Pr(\omega=\omega^H)=\mu_0\in(0,1)$.
$F$ is the CDF with density $f$, where $F(0)=0$, and $F(x)$ is twice-differentiable and $f(x)>0$ for $x\in\left(0,c_c/(1-p)\right)$.
We denote a posterior belief after a signal  $s$ by $\mu_s=\Pr(\omega=\omega^H|s)$.
We denote by $\widehat v^b(\mu_s,s)$ and $\widehat v^t(\mu_s,s)$ the defender's expected (interim) payoff, given a signal $s$ with the associated posterior $\mu_s$, in the burning-money and audience-cost settings, respectively.

The game without commitment proceeds as follows:
\begin{enumerate}
    \item Nature determines $\omega$ and $v_c$.
    \item The defender observes $\omega$ and chooses $s$.
    \item The challenger observes $v_c$ and $s$ and chooses $a$.
    \item The defender observes $a$ and chooses $b$.
    \item Payoffs are realized, and the game ends.
\end{enumerate}
 
In the burning money setting, we further analyze an environment in which the defender, before learning the state $\omega$, commits to a strategy of choosing $s$. That is, the defender, before observing $\omega$, commits to a strategy $\pi(s|\omega):\{\omega^L,\omega^H\}\rightarrow\Delta(S)$ that maps $\omega$ to a probability distribution over $S$.\footnote{Commitment power in audience-cost settings seems substantively implausible. For theoretical completeness, the results for the audience-cost setting with commitment are in Proposition \ref{Prop_AC_2sided} in  Appendix \ref{app:C}.} The timing of the game in this setting is as follows:
\begin{enumerate}
\item The defender commits to $\pi(s|\omega)$ and this mapping is observed by the challenger.
    \item Nature determines $\omega$ and $v_c$.
    \item The defender observes $\omega$ and chooses $s$ according to $\pi(s|\omega)$.
    \item The challenger observes $v_c$ and $s$ and chooses $a$.
    \item The defender observes $a$ and chooses $b$.
    \item Payoffs are realized, and the game ends.
\end{enumerate}

We focus on sender-optimal perfect Bayesian equilibria that satisfy the D1 criterion \citep{banksEquilibriumSelectionSignaling1987}. Given parameters, a sender-optimal equilibrium is one that maximizes the sender's (defender's) ex ante payoff.

The following notation will simplify the exposition: \begin{equation*}
    W^i=p\omega^i-c_d,\, i\in\{L,H\},\, \delta=\omega^H-W^H,\, v^*_c(\mu)=\frac{\mu c_c}{1-\mu p},\, F_{\mu}= F(v_c^*(\mu)),\end{equation*}
where $\mu$ is the challenger's belief that $\omega=\omega^H$. $W^i$ is a type $i$ defender's payoff from war, and $\delta$ is the high type's expected marginal gain from deterring the challenger over going to war with the challenger. If $W^L,\, W^H>0$, neither signaling through burning money nor creating audience costs has any value. If  $W^L,\, W^H<0$, signaling through burning money has no value and hence the comparison is straightforward. Thus, we focus on settings in which 
\begin{equation*}
    W^L<0<W^H,
\end{equation*} 
so that, under the default of $s=0$ in \eqref{eqn:v-b}-\eqref{eqn:v-t}, a high-type defender will defend if attacked, but the low-type defender will not. 

We collect the assumptions used in the analysis below. We specify the parts that are used in each proposition.

\begin{assumption}
    \begin{enumerate}[label={\textup{(\roman*)}},ref=\roman*]
        \item\label{assum_convex} The defender's interim payoff in the burning-money setting, $\widehat v^b(\mu,s)$ in equation \eqref{payoff_uncertain_BM}, is strictly convex in $\mu$.
        \item\label{assum_sb1} The smallest strictly positive signal in the burning-money setting is sufficiently small:     
        \begin{empheq}[left={s^b_1 < \empheqlbrace}]{align}
            &F_1\, \omega^L \tag{ii-a}\label{assum_sb1_F1}\\
            & \left(F_1-F_{0.5}\right)\delta-F_{0.5}\ \omega^L \tag{ii-b}\label{assum_sb1_1/2}\\
            &\left(F_1-\frac{\overline{s}^t}{\omega^L+\overline{s}^t}\right)\, \delta\tag{ii-c}\label{assum_sb1_sbart}
        \end{empheq}
        \item \label{assum_rich} The set of feasible signals in the burning-money setting is rich in the following sense:
        \begin{align}
            & F_1\,\omega^L\in S^b \tag{iii-a}\label{assum_rich_F1}\\
            & S^b\cap \left(s,s+(F_1-F_{\mu_0}) \delta\right)\neq\emptyset~\text{for any }s\in S^b\cap\left[0,F_{\mu_0}\, \omega^L\right] \tag{iii-b}\label{assum_rich_s+Delta}
        \end{align}
    \end{enumerate}
    \label{assum_2sided}
\end{assumption}

Assumption \ref{assum_2sided}-(\ref{assum_convex}) is satisfied with common distributions of the challenger's valuation $v_c$ such as uniform, exponential, and lognormal in natural settings where the challenger's expected valuation normalized by the cost of war, $\mathbb{E}[v_c]/c_c$, is sufficiently large---see Proposition \ref{prop_convex} in the Appendix. 

The remaining assumptions are technical, concerning the richness of the action (signal) space in the burning-money setting $S^b$. We state them to invoke the exact requirements in each result; \ref{assum_2sided}-(\ref{assum_sb1}) means that small actions (signals) are feasible. The right-hand side is strictly positive: $\omega^L,F_1>0$; \ref{assum_2sided}-(\ref{assum_sb1_1/2}) is only invoked together with  \ref{assum_2sided}-(\ref{assum_convex}), under which the right-hand side of (\ref{assum_sb1_1/2}) is strictly positive (see Remark \ref{remark_v-c_positive} in the Appendix);
 \ref{assum_2sided}-(\ref{assum_sb1_sbart}) is only invoked when $\overline s^t$ is sufficiently small that the right-hand side is strictly positive.
 \ref{assum_2sided}-(\ref{assum_rich}) ensures that the results are not driven by the coarseness of the action (signal) space; e.g., (\ref{assum_rich_F1}) ensures that the signal that maximizes the defender's payoff among separating equilibria is feasible.

\section{Burning Money and the Reversal of Its Meaning}

We begin with the burning-money setting. Let $\widehat a(\mu_s,s,v_c)$ and $\widehat b(\omega,s)$ be the equilibrium values of $a$ and $b$, respectively. From \eqref{eqn:v-b} and $W^L<0<W^H$, the defender defends against an attack if and only if $\omega=\omega^H$: \vspace{-1.5em}
\begin{equation}\label{eqn:D-defend}
    \widehat b(\omega,s)=\mathbbm{1}(\omega=\omega^H).
\end{equation}
Anticipating this, from \eqref{eqn:u}, a challenger with a belief $\mu$ attacks if and only if $v_c>v^*_c(\mu)$. Therefore, from the perspective of the defender at the beginning of the game, upon observing a signal $s$, the challenger will attack with probability
\begin{equation}
1-F_{\mu_s},\label{eqn:C-prob-attack}
\end{equation}
where $\mu_s$ is the challenger's belief upon observing signal $s$. Combining \eqref{eqn:v-b}, \eqref{eqn:D-defend}, and \eqref{eqn:C-prob-attack}, the defender's interim payoff is
\begin{eqnarray}\label{payoff_uncertain_BM}
\widehat v^b(\mu_s,s) &=& \mathbb{E}_{\mu_s,v_c}\left[v^b(s,\widehat a(\mu_s,s,v_c),\widehat b(\omega, s),\omega)\right]\nonumber\\
&=& -s+\left(1-F_{\mu_s}\right)\ W^H\ \mu_s+F_{\mu_s} (\mu_s \omega^H+(1-\mu_s)\omega^L).
\end{eqnarray}

\begin{prop}[Burning Money without Commitment]\label{benchmark_BM2}
Suppose Assumption \ref{assum_2sided}-(\ref{assum_rich_F1}) and (\ref{assum_rich_s+Delta}) hold. Consider the burning-money setting without commitment.  In equilibrium, the defender chooses $s=F_1\omega^L$ when $\omega=\omega^H$ and $s=0$ when $\omega=\omega^L$; the challenger attacks when $s<F_1\omega^L$ or $v_c>v^*_c(1)$, upon which the defender fights if and only if $\omega=\omega^H$.
    
\end{prop}

Proposition \ref{benchmark_BM2} shows that burning-money settings without commitment lead to separation, in which high types choose the least costly action that prevents low types from imitating. It is a canonical benchmark that appears in different forms in the literature \citetext{\citealp[pp.~75--76]{fearonSignalingForeignPolicy1997}; \citealp[p.~17]{reichWhenCanStates}; \citealp[p.~33]{slantchev2011book}}.
Other equilibria where low types pool with high types are not sustainable by D1 refinement, which requires the challenger to have reasonable beliefs off the equilibrium path, and by the sender optimal equilibrium selection, so that the low type separates when indifferent.\footnote{This latter choice is a technical assumption with no substantive consequences when the action (signal) space is sufficiently rich. Without it, high types would choose a larger signal than $F_1\omega^L$ to separate themselves.} 

The equilibrium minimizes the chances of war by deterring all deterrable challengers when the defender is of high type, and through concession when the defender is of low type. However, information transmission is costly as high types must burn just enough money to preclude imitation by low types. 

The cost of this approach stems from the need for the high type to take sufficiently costly action. Our key observation is that if the defender could, ex ante, before learning its type, commit to a burning-money strategy, it could maintain separation while reducing the signaling costs of information transmission. This commitment amounts to the defender ex ante choosing a mapping $\pi(s|\omega)$ from its type $\omega\in\{\omega^L,\omega^H\}$ into a probability distribution over its actions $S^b$.

Figure \ref{fig:burning-money-with-commitment} illustrates the logic. Dashed lines are the defender's interim payoffs $\hat v^b(\mu,s)$ for $s\in S^b=\{0,0.1,0.15\}$. $\hat v^b(\mu,s)$ is the defender's payoff if it takes action $s$ and the belief that it is a high type is $\mu$. Of course, not all such beliefs are feasible: a rational challenger makes inferences about the defender's type upon observing the defender's action $s$. The defender, of course, would like to somehow induce the challenger to believe that it is a high type by doing the free ``action'' of doing nothing $s=0$---it wants to be at the most extreme northeast corner: $\widehat v^b(\mu=1,s=0)$. But it cannot: If the defender always does nothing, the challenger maintains the prior $\mu_0$. What can the defender achieve and how? Proposition \ref{Prop_BM_2sided} describes the answer.

\begin{prop}[Burning Money with Commitment]\label{Prop_BM_2sided}
     Suppose Assumption \ref{assum_2sided}-(\ref{assum_convex}), (\ref{assum_sb1_F1}), and (\ref{assum_sb1_1/2}) hold. Consider the burning-money setting with commitment. In equilibrium,
    \begin{itemize}
        \item If $\mu_0\le1/2$, the defender sends $s=s^b_1$ when $\omega=\omega^H$ and $s=0$ when $\omega=\omega^L$. The challenger challenges if and only if $s=0$ or $v_c>v^*_c(1)$, upon which the defender fights if and only if $\omega=\omega^H$.
        \item If $\mu_0\ge1/2$, the defender sends $s=0$ when $\omega=\omega^H$ and $s=s^b_1$ when $\omega=\omega^L$. The challenger challenges if and only if $s=s^b_1$ or $v_c>v^*_c(1)$, upon which the defender fights if and only if $\omega=\omega^H$.
    \end{itemize}
    For any $\mu_0\in(0,1)$, commitment power strictly improves the defender's ex ante payoff.  \end{prop}

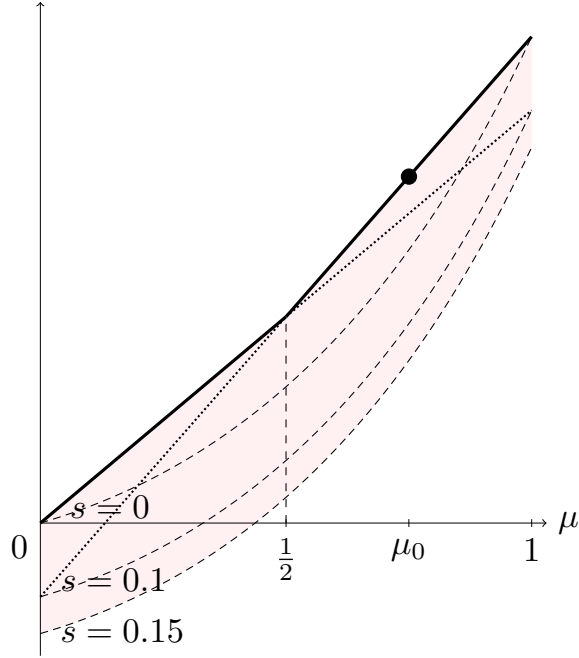
\begin{figure}[t]
\centering\footnotesize
\begin{tikzpicture}[scale=1.3,transform shape]
    \fill[
        pink!40,
        opacity=0.55,
        scale=5
    ]
    plot[
        domain=0:0.5,
        smooth,
        variable=\x
    ]
    ({\x},
    {1.5*\x*(0.4*1.5-0.5
    +(0.6/(2.5*(1-0.4)))
      *((1-0.4)*1.5+0.5)-0.1)})
    --
    plot[
        domain=0.5:1,
        smooth,
        variable=\x
    ]
    ({\x},
    {1.5*(\x*(0.4*1.5-0.5
    +(0.6/(2.5*(1-0.4)))
      *((1-0.4)*1.5+0.5)+0.1)-0.1)})
    --
    plot[
        domain=1:0,
        smooth,
        variable=\x
    ]
    ({\x},
    {1.5*(\x*(0.4*1.5-0.5
    +((\x*0.6)/(2.5*(1-(\x*0.4))))
      *((1-0.4)*1.5+0.5))
    +(1-\x)*
      (((\x*0.6)/(2.5*(1-(\x*0.4))))*0.4)-0.15)})
    -- cycle;

    \draw[very thin,->]
        (0,0) -- (5.15,0) node[right] {$\mu$};

    \draw[very thin,->]
        (0,-1.35) -- (0,5.3);

    \node[below left] at (0,0) {$0$};

    \draw[very thin]
        (2.5,0.04) -- (2.5,-0.04)
        node[below] {$\frac12$};

    \draw[very thin]
        (3.75,0.04) -- (3.75,-0.04)
        node[below] {$\mu_0$};

    \draw[very thin]
        (5,0.04) -- (5,-0.04)
        node[below] {$1$};

    \draw[
        black,
        thin,
        densely dashed,
        scale=5,
        domain=0:1,
        smooth,
        variable=\x
    ]
    plot ({\x},
    {1.5*(\x*(0.4*1.5-0.5
    +((\x*0.6)/(2.5*(1-(\x*0.4))))
      *((1-0.4)*1.5+0.5))
    +(1-\x)*
      (((\x*0.6)/(2.5*(1-(\x*0.4))))*0.4))});

    \draw[
        black,
        thin,
        densely dashed,
        scale=5,
        domain=0:1,
        smooth,
        variable=\x
    ]
    plot ({\x},
    {1.5*(\x*(0.4*1.5-0.5
    +((\x*0.6)/(2.5*(1-(\x*0.4))))
      *((1-0.4)*1.5+0.5))
    +(1-\x)*
      (((\x*0.6)/(2.5*(1-(\x*0.4))))*0.4)-0.1)});

    \draw[
        black,
        thin,
        densely dashed,
        scale=5,
        domain=0:1,
        smooth,
        variable=\x
    ]
    plot ({\x},
    {1.5*(\x*(0.4*1.5-0.5
    +((\x*0.6)/(2.5*(1-(\x*0.4))))
      *((1-0.4)*1.5+0.5))
    +(1-\x)*
      (((\x*0.6)/(2.5*(1-(\x*0.4))))*0.4)-0.15)});

    \draw[
        very thick,
        black,
        scale=5,
        domain=0:0.5,
        smooth,
        variable=\x
    ]
    plot ({\x},
    {1.5*\x*(0.4*1.5-0.5
    +(0.6/(2.5*(1-0.4)))
      *((1-0.4)*1.5+0.5)-0.1)});

    \draw[
        very thick,
        black,
        scale=5,
        domain=0.5:1,
        smooth,
        variable=\x
    ]
    plot ({\x},
    {1.5*(\x*(0.4*1.5-0.5
    +(0.6/(2.5*(1-0.4)))
      *((1-0.4)*1.5+0.5)+0.1)-0.1)});

    \draw[
        thin,
        black,
        dashed
    ]
    (2.5,0) -- (2.5,2.1);

    \draw[
        thick,
        black,
        densely dotted
    ]
    (0,-0.75) -- (2.5,2.1);

    \draw[
        thick,
        black,
        densely dotted
    ]
    (2.5,2.1) -- (5,4.2);

    \coordinate (prior) at (3.75,3.525);

    \filldraw[black]
        (prior) circle (2.2pt);

    \node[anchor=west] at (0.18,0.16)
        {$s=0$};

    \node[anchor=west] at (0.08,-0.58)
        {$s=0.1$};

    \node[anchor=west] at (0.08,-1.1)
        {$s=0.15$};

\end{tikzpicture}

\caption{Burning Money with Commitment}
\label{fig:burning-money-with-commitment}
\end{figure}

Proposition \ref{Prop_BM_2sided} shows that the defender continues to separate with commitment power but uses very different signals to do so. When the high type is unlikely, as before, the low type does nothing, while the high type takes a costly action. However, this action is the least costly positive action. It no longer needs to be sufficiently costly to prevent imitation by the low type. It only needs to be an action different from the free ``doing nothing'' action of the low type. This already saves the defender the cost of money burning. 

Remarkably, when the high type is likely, the defender's behavior changes starkly. Now, the high type will take no action, while the low type takes a costly action. That is, burning money now reveals that the defender has low resolve. Intuitively, if the defender wants separation, it only needs to ensure that different types take different actions. To save money, these actions should be either to do nothing or to take the least costly action. But which type should take which action? The type that is ex ante less likely should take the more expensive action, and the type that is ex ante more likely should take the less expensive action.

Proposition \ref{Prop_BM_2sided} applies the techniques developed in \citet{BoleslavskyShadmehr_Signaling_with_Commitmentpdf}. We now describe the logic of this result in more detail. Recall that the defender would like to burn no money $s=0$ and at the same time induce the challenger to believe it is a high type. But it cannot; e.g., if all types always do nothing, then no information is transmitted, and the challenger's belief remains at the prior. What restrictions does the defender have, and how can it maximize its payoffs subject to those restrictions?

Adopting the analyses of \citet{kamenicaBayesianPersuasion2011} and \citet{BoleslavskyShadmehr_Signaling_with_Commitmentpdf}, the defender faces two simple restrictions. First, beliefs and expectations must satisfy the law of iterated expectations. In particular, if action (signal) $s$ induces a belief $\mu_s$ and if action $s$ is taken with probability $\tau_s$, then: $\mathbb{E}_{\tau}[\mu_s]=\mathbb{E}_{\tau}[ \mathbb{E}[\mathbf{1}\{\omega=\omega^H\}|s]  ]=\mathbb{E}_\omega[\mathbf{1}\{\omega=\omega^H\}]=\mu_0$. Because action $s$ induces $\mu_s$, $(\tau_s)_{s\in S^b}$ is effectively a distribution over posteriors. Moreover, any combination of posteriors that satisfies the law of iterated expectations can be induced by an appropriately chosen mapping $\pi(s|\omega)$. In particular, by (ex ante) committing to a signaling strategy $\pi(s|\omega):\{\omega^L,\omega^H\}\rightarrow \Delta\left(S^b\right)$, the defender can ``split'' any prior belief of the challenger $\mu_0$ into several posterior beliefs:
\begin{equation}\label{eqn:bayes-plausibility}
    \sum_{s\in S^b}\tau_s\, \mu_s=\mu_0\ \ \mbox{s.t.}\ \ \sum_{s\in S^b}\tau_s=1\ \ \mbox{and}\ \  \tau_s\geq0,
\end{equation}
where $\tau_s$ here is the probability that action (signal) $s\in S^b$ is sent, and hence the posterior $\mu_s$ is induced.
Second, each action $s\in S^b$ can induce only one posterior belief. It follows that the defender can produce any convex combination of its interim payoffs $\{\widehat v^b(\mu_s,s)\, :\, s\in S^b\}$ that satisfies the law of iterated expectations:
\begin{equation}\label{eqn:defender-payoff-in-signaling-w-commitment}
    \max_{(\mu_s,\tau_s)_{s\in S^b}}\sum_{s\in S^b} \tau_s\, \widehat v^b(\mu_s,s)\ \ \mbox{s.t.}\ \ \eqref{eqn:bayes-plausibility}.
\end{equation}
The optimization \eqref{eqn:defender-payoff-in-signaling-w-commitment} is simply the convex combination of the set of all defender's interim payoffs with the restriction that each interim payoff must be used at most once: a convex combination cannot use two points from the same interim payoff graph, because that would imply that the same action (signal) induces two different posterior beliefs. This is the topological join of a set of interim payoff graphs, $\{\widehat v^b(\mu_s,s)\, :\, s\in S^b\}$. The shaded area in Figure \ref{fig:burning-money-with-commitment} shows the join. At any given prior $\mu=\mu_0$, the defender wants to create the maximum value of the join, the join envelope, shown in the figure by solid black lines. We note that, in contrast to the Bayesian persuasion literature, in which the convex hull and concave envelope are the relevant objects, the join set is not generally convex and the join envelope is not generally concave---e.g., it is convex here.

Figure \ref{fig:burning-money-with-commitment} provides a simple characterization of the defender's optimal commitment strategy. Reading off the join envelope, the solid black line, when the high type is ex ante unlikely ($\mu_0< 1/2$), the defender splits the prior into two posteriors $\mu=0$ and $\mu=1$: when $s=0$, the posterior is $\mu_{s=0}=0$; when $s=0.1$, it is $\mu_{s=0.1}=1$. That is, the low type burns no money, and the high type burns the minimum available amount $s^b_1=0.1$. There is full separation, but at a minimal cost: because the defender has committed to the mapping $\pi(s=0|\omega=\omega^L)=\pi(s=0.1|\omega=\omega^H)=1$,  the low type is not able to imitate the high type. The cost of signaling here stems from the need to take an action different from the free action of the low type, not from the need to take an expensive enough action that prevents the low type's imitation.

By contrast, when the high type is ex ante likely ($\mu_0> 1/2$), the defender still splits the prior into two posteriors $\mu=0$ and $1$, but reverses the signals used to induce them: when $s=0$, the posterior is $\mu_{s=0}=1$, not $0$; when $s=0.1$, it is $\mu_{s=0.1}=0$, not $1$. That is, the high type burns no money, and the low type burns the minimum available amount $s^b_1=0.1$. Because the defender has committed to the mapping $\pi(s=0.1|\omega=\omega^L)=\pi(s=0|\omega=\omega^H)=1$, the low type cannot imitate the high type. The defender could still employ the strategy of having high types burn money, while low types do nothing. That would still save the defender money compared to the setting without commitment in which the high types would take even more costly actions. But the defender can do better by reversing the meaning of the actions: burning money means low type, while doing nothing means high type. 
The low type will have strictly negative payoffs: they take a costly action and back down. But they are unlikely to arise, and the costs saved by the likely high types more than compensate for those losses. 

Through the reversal of the meaning of burning money, the defender can significantly reduce the costs of the burning-money approach and raise its payoff. As Figure \ref{fig_comparison} demonstrates, this commitment power generates most value exactly when the high type is more likely, and hence burning money without commitment, or even with commitment but without reversal of meaning, is most expensive.

\section{Audience Costs}

We now analyze the audience-cost setting and then compare the defender's ex ante payoffs in the two settings. We distinguish between the settings in which audience costs can be large and those where they cannot.

Upon the challenger's attack, the defender fights if $\omega=\omega^H$ or $-s\leq  p\omega^L-c_d$. Thus, the defender's interim payoff is
{\small
\begin{eqnarray}\label{payoff_uncertain_AC}
&&\widehat v^t(\mu_s,s)\\
&=&\mathbbm{1}(-s>  p\omega^L-c_d) \left[ F\left(v^*_c(\mu_s)\right) (\mu_s \omega^H+(1-\mu_s)\omega^L)+\left(1-F\left(v^*_c(\mu_s)\right)\right) (\mu_s(p\omega^H-c_d)-(1-\mu_s)s) \right] \notag\\
&&+\ \ \mathbbm{1}(-s\leq  p\omega^L-c_d)\ \left[F_1\ (\mu_s \omega^H+(1-\mu_s)\omega^L)+\left(1-F_1\right)\ (\mu_s(p\omega^H-c_d)+(1-\mu_s)(p\omega^L-c_d)) \right].\notag
\end{eqnarray}
}
\normalsize

\subsection{Large Audience Costs}

Consider settings where sufficiently large audience costs are feasible:
$$\overline{s}^t\ge s^*:=\min\left\{ -p\omega^L+c_d, \frac{F_1}{1-F_1}\, \omega^L \right\},$$
where we recognize that $s^*$ is a function of $p,\ c_d,\ \omega^L$, and $F_1$. We use the convention that $\pm a/0=\pm\, \infty$ for $a>0$.  Moreover, $s^*$ is weakly increasing in $F_1$, the likelihood that the challenger does not attack if it is sure that the defender will defend, with $\lim_{F_1\rightarrow 0}s^*=0$ and $\lim_{F_1\rightarrow1}s^*=-p\omega^L+c_d>0$. In particular, if the challenger attacks almost surely even if it knows that the defender will defend upon an attack ($F_1\approx0$), then we are in this setting.

\begin{prop}[Large Audience Costs]\label{AC2_Large}
  Suppose $\overline{s}^t\ge s^*$. In the audience-cost setting, 
    \begin{itemize}
        \item If $-p\omega^L+c_d
        > \frac{F_1}{1-F_1}\, \omega^L$, the defender sends some $s\ge\frac{F_1}{1-F_1}\omega^L$ when $\omega=\omega^H$ and sends $s=0$ when $\omega=\omega^L$; the challenger challenges when $s<\frac{F_1}{1-F_1}\omega^L$ or $v_c>v^*_c(1)$, upon which the defender fights if and only if $\omega=\omega^H$.
        \item If $-p\omega^L+c_d
        < \frac{F_1}{1-F_1}\, \omega^L$, each type of the defender sends a signal $s\ge c_d-p\omega^L$. The challenger challenges when $s<c_d-p\omega^L$ or $v_c>v^*_c(1)$. Both types of the defender fight when challenged on the equilibrium path.
    \end{itemize}
    \label{benchmark_AC2_BridgeBurning}
\end{prop}

We note that in the second equilibrium, the two types of the defender do not have to send the same $s\ge c_d-p\omega^L$. Even if they send different $s^\prime$ and $s^{\prime\prime}\neq s^\prime$, these signals constitute an equilibrium that induces the same probability of attack and war, as long as $s^\prime,s^{\prime\prime}\ge c_d-p\omega^L$. 

Intuitively, the low type can choose $s\geq |p\omega^L-c_d|$ to commit itself to fighting if attacked, thereby transforming itself into the high type from the challenger's perspective. The low type will do that whenever
\begin{equation}\label{eqn:AC-low-type-commitment}
    F_1\omega^L+(1-F_1)(p\omega^L-c_d)\geq0,\ \mbox{i.e.}\ \frac{F_1}{1-F_1}\omega^L\geq -p\omega^L+c_d,
\end{equation}
where $0$ is the low type's payoff if the challenger attacks and the defender does not defend. This logic has nothing to do with information transmission, and it is exactly the second equilibrium in Proposition \ref{AC2_Large}. When \eqref{eqn:AC-low-type-commitment} is reversed, the weak type does not want to commit. Instead, it wants to imitate the high type and reduce the probability of attack to $1-F_1$, but backs down if the attack actually happens. But the high type can always choose a high enough signal to make imitation unattractive to the low type. If the high type takes $s\geq |p\omega^L-c_d|$, the low type would not imitate. In fact, even if the high type takes $s\geq \frac{F_1}{1-F_1}\omega^L$, the low type will not imitate because
\begin{equation*}
    F_1\omega^L+(1-F_1)(-s)\leq  F_1\omega^L+(1-F_1)\left(-\frac{F_1}{1-F_1}\omega^L\right)=0.
\end{equation*}

That is, if $\overline s^t\geq \frac{F_1}{1-F_1}\omega^L$, the high type can choose a signal that the low type would want to imitate if and only if it wants to follow through and defend against an attack. Importantly, the low type would want to do so exactly when \eqref{eqn:AC-low-type-commitment} holds. But then, the low type would not need to pretend to be the high type; it could commit on its own by taking $s>|p\omega^L-c_d|$.

With full separation or the low type's commitment to defend if attacked, it follows that the defender's ex ante payoff (under the prior) is linear in the prior $\mu_0$. The red line in Panel (a) of Figure \ref{fig_comparison} demonstrates this.

\begin{figure}[t]
\centering
\footnotesize

\begin{minipage}[t]{0.5\textwidth}
\centering

\begin{tikzpicture}[x=1.15cm,y=1cm]

    \draw[very thin,->] (0,0) -- (5.15,0) node[right] {$\mu_0$};
    \draw[very thin,->] (0,-0.8) -- (0,5.3);

    \node[below left] at (0,0) {$0$};

    \draw[very thin] (5,0.04) -- (5,-0.04)
        node[below] {$1$};

    \draw[very thick,red,scale=5,domain=0:1,
          variable=\x]
    plot ({\x},
        {0.75*\x*(0.6*3-1.5
        +(0.4/(2.5*(1-0.6)))*((1-0.6)*3+1.5))});
    \draw[very thick,blue,scale=5,domain=0:49/51,
          variable=\x]
    plot ({\x},
        {0.75*\x*(0.6*3-1.5
        +(0.4/(2.5*(1-0.6)))*((1-0.6)*3+1.5-0.8))});

    \draw[very thick,black,scale=5,domain=0:0.5,
          variable=\x]
    plot ({\x},
        {0.75*\x*(0.6*3-1.5
        +(0.4/(2.5*(1-0.6)))*((1-0.6)*3+1.5)-0.1)});
    \draw[very thick,black,scale=5,domain=0.5:1,
          variable=\x]
    plot ({\x},
        {0.75*(\x*(0.6*3-1.5
        +(0.4/(2.5*(1-0.6)))*((1-0.6)*3+1.5)+0.1)-0.1)});

\end{tikzpicture}

\par\vspace{2pt}
\textbf{(a) Large audience costs}

\end{minipage}
\hfill
\begin{minipage}[t]{0.49\textwidth}
\centering

\begin{tikzpicture}[x=1.15cm,y=1.02cm]

    \draw[very thin,->] (0,0) -- (5.15,0) node[right] {$\mu_0$};
    \draw[very thin,->] (0,-0.55) -- (0,5.3);

    \node[below left] at (0,0) {$0$};

    \draw[very thin] (5,0.05) -- (5,-0.05)
        node[below] {$1$};

    \draw[very thick,red,scale=5,domain=0:9/11,
          variable=\x]
    plot ({\x},
        {0.5*\x*(0.6*3-1.5
        +(0.6/(0.8+0.6))*((1-0.6)*3+1.5))});
    \draw[very thick,red,scale=5,domain=9/11:1,samples=100,
          variable=\x]
    plot ({\x},
        {0.5*(\x*(0.6*3-1.5
        +((\x*0.4)/(1.5*(1-(\x*0.6))))*((1-0.6)*3+1.5))
        +(1-\x)*(((\x*0.4)/(1.5*(1-(\x*0.6))))*0.8
        -(1-((\x*0.4)/(1.5*(1-(\x*0.6)))))*0.6))});
    \draw[very thick,blue,scale=5,domain=0:85/87,
          variable=\x]
    plot ({\x},
        {0.5*\x*(0.6*3-1.5
        +(0.4/(1.5*(1-0.6)))*((1-0.6)*3+1.5-0.8))});

    \draw[very thick,black,scale=5,domain=0:0.5,
          variable=\x]
    plot ({\x},
        {0.5*\x*(0.6*3-1.5
        +(0.4/(1.5*(1-0.6)))*((1-0.6)*3+1.5)-0.1)});
    \draw[very thick,black,scale=5,domain=0.5:1,
          variable=\x]
    plot ({\x},
        {0.5*(\x*(0.6*3-1.5
        +(0.4/(1.5*(1-0.6)))*((1-0.6)*3+1.5)+0.1)-0.1)});
    \draw[thin,black,dashed]
        ({5*((-5+sqrt(10105))/112)},0) --
        ({5*((-5+sqrt(10105))/112)},{5*(0.5*((-5+sqrt(10105))/112)*(0.6*3-1.5
        +(0.4/(1.5*(1-0.6)))*((1-0.6)*3+1.5-0.8)))});
    \node[below] at ({5*((-5+sqrt(10105))/112)},0) {$\overline\mu_0$};    
\end{tikzpicture}

\par\vspace{2pt}
\textbf{(b) Small audience costs}

\end{minipage}

\caption{The defender's ex ante payoff under audience costs (red), burning
money without commitment (blue), and burning money with commitment (black).
In both panels, $\omega^L=0.8$, $\omega^H=3$, $p=0.6$, $c_d=1.5$,
$c_c=0.4$, and
$S^b=S^t=\{0,0.1,
\dots, 0.6\}$,
so $s_1^b=0.1$ and $\overline{s}^t=0.6$.
In panel (a), $v_c\sim U[0,2.5]$, $F_1=\frac25$, and $s^*=\frac{8}{15}$.
In panel (b), $v_c\sim U[0,1.5]$, $F_1=\frac23$, and
$s^*=\frac{51}{50}>\overline{s}^t$.
Vertical scales differ across panels.}
\label{fig_comparison}
\end{figure}
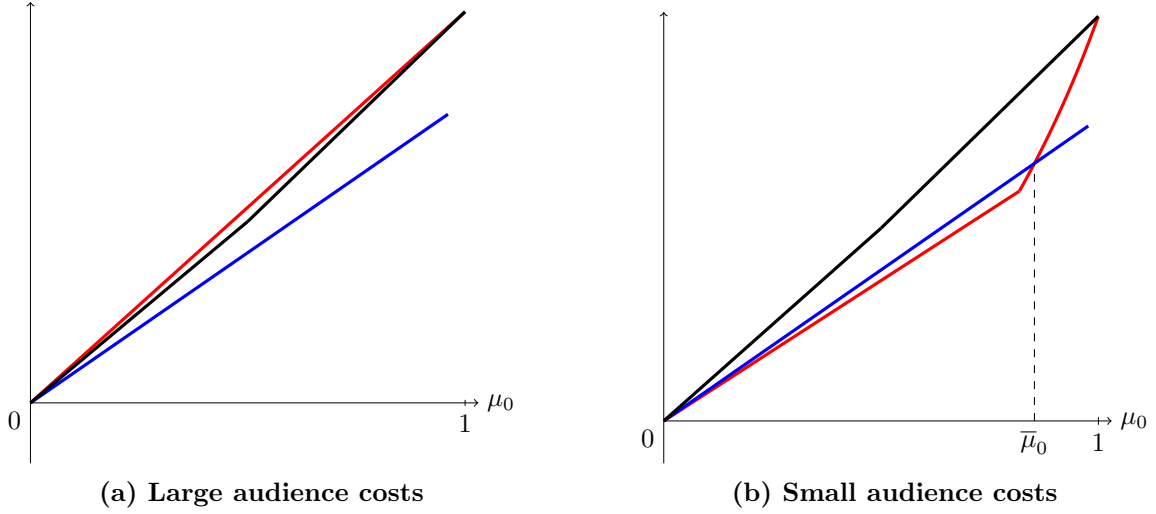

Thus, in large audience-cost settings where $\overline s^t\geq s^*$, signaling has value if and only if $\frac{F_1}{1-F_1}\, \omega^L< -p\omega^L+c_d$. Otherwise, the commitment of burning bridges alone, without any signaling value, delivers the results. Moreover, when $\frac{F_1}{1-F_1}\, \omega^L< -p\omega^L+c_d$, so that signaling has value in the audience-cost setting, combining  Propositions \ref{benchmark_BM2} and \ref{AC2_Large} reveals that signaling requires a more costly available signal in the audience-cost setting, $\frac{F_1}{1-F_1}\omega^L$, than in the burning-money setting where $F_1\omega^L<\frac{F_1}{1-F_1}\omega^L$. 

That is, while signaling costs are not paid in equilibrium in large audience-cost settings, they require a higher feasibility threshold: the same cost $F_1\omega^L$ that delivers separation in burning-money settings will not deliver separation in audience-cost settings.

Combining Propositions \ref{benchmark_BM2}, \ref{Prop_BM_2sided}, and \ref{AC2_Large} allows us to compare the defender's ex ante payoffs across the audience-cost and burning-money settings.

\newpage
\begin{prop}[Large Audience Costs versus Burning Money]\label{Comparison_BMwoComm_ACwoComm2}
    Suppose Assumption \ref{assum_2sided}-(\ref{assum_convex}), (\ref{assum_sb1_F1}), and (\ref{assum_sb1_1/2}) hold and  $\overline{s}^t\ge s^*$. For any $\mu_0\in(0,1)$, the defender's ex ante payoff is strictly higher in the audience-cost setting than in the burning-money setting with and without commitment.
\end{prop}

Without commitment, this result echoes the result in \citet{fearonSignalingForeignPolicy1997}. Moreover, the superiority of audience cost remains even when the defender has commitment power. That is, although Proposition \ref{Prop_BM_2sided} shows that commitment power strictly improves the defender's ex ante payoff from burning money, it is still strictly lower than that from signaling by audience cost.

\subsection{Small Audience Costs}

If the defender can create a sufficiently large audience cost, it achieves the best possible payoff: all deterrable challengers will be deterred at no realized cost. However, the literature highlights that the ability to create large audience costs is rare \citep{katagiriCredibilityPublicPrivate2019,snyderCostEmptyThreats2011,trachtenbergAudienceCostsHistorical2012}. We thus consider possibly more realistic settings where the ability to create audience costs is limited:
\begin{equation*}
    \overline{s}^t<s^*.
\end{equation*}
 Proposition \ref{AC2_Limitation} characterizes the equilibrium.

\begin{prop}[Small Audience Costs]\label{AC2_Limitation}
Suppose $\overline{s}^t<s^*$. In the audience cost setting,
    \begin{itemize}
        \item If $\frac{F_{\mu_0}}{1-F_{\mu_0}}\,\omega^L<\overline{s}^t<s^*$, the defender sends $s=\overline{s}^t$ when $\omega=\omega^H$. He mixes $s=\overline{s}^t$ and $s=0$ when $\omega=\omega^L$. The challenger challenges when $s<\overline{s}^t$ or $v_c>v^*_c\left(\mu_{s=\overline{s}^t}\right)$, upon which the defender fights if and only if $\omega=\omega^H$.
        \item If $\overline{s}^t\le\frac{F_{\mu_0}}{1-F_{\mu_0}}\,\omega^L$, both types of the defender send $s=\overline{s}^t$. The challenger challenges when $s<\overline{s}^t$ or $v_c>v^*_c(\mu_0)$, upon which the defender fights if and only if $\omega=\omega^H$. 
    \end{itemize}
\end{prop}

With limited audience costs, the low type cannot commit to defending if attacked, and the high type cannot fully separate. The low type tries to bluff, pretending to be a high type, but backs down when the challenger attacks. As the maximum feasible audience cost falls, equilibrium behavior moves from fully separating to semi-separating to pooling.

As the maximum feasible audience cost falls, less
information is transmitted. However, the defender's payoff is generally non-linear. To see this, suppose $\overline{s}^t<\frac{F_{\mu_0}}{1-F_{\mu_0}}\,\omega^L$, so that all defender types pool on $\overline s^t$, and the low type will back down if attacked. Conditional on pooling behavior, higher audience costs have no informational effect and only reduce the defender payoffs, as they will be incurred by low types. By contrast, when $\overline s^t$ is large enough that the equilibrium is semi-separating, higher feasible audience costs reduce the probability of imitation, raising deterrence and the defender's ex ante payoffs.\footnote{We note that when mixing, the low type remains indifferent between imitating the high type and sending no signal, so its payoff does not change. But the high type benefits from the additional deterrence.} When equilibrium behavior moves from semi-separating to pooling as the feasible audience cost falls, the informational channel and the direct channel coexist and push the relative payoffs in opposite directions.

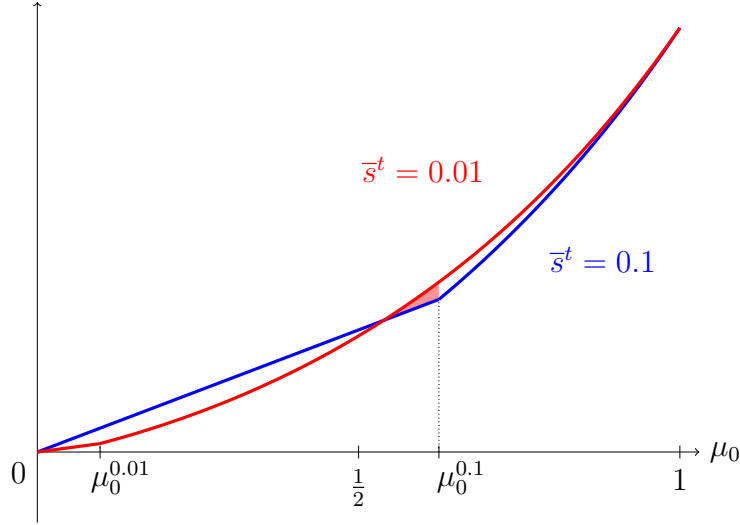
\begin{figure}[t]
\begin{center}
\begin{tikzpicture}[x=1.7cm,y=1.7cm]

    \draw[very thin,->] (0,0) -- (5.15,0) node[right] {$\mu_0$};
    \draw[very thin,->] (0,-0.55) -- (0,3.5);

    \node[below left] at (0,0) {$0$};

    \draw[very thin] (2.5,0.05) -- (2.5,-0.05)
        node[below] {$\frac{1}{2}$};

    \draw[very thin] (5,0.05) -- (5,-0.05)
        node[below] {$1$};

    \fill[red!70,opacity=0.6,scale=5]
    plot[
        domain=0.5386697094:5/8,
        samples=120,
        smooth,
        variable=\x
    ]
    ({\x},
    {\x*(0.4*1.5-0.5
    +((\x*0.6)/(2.5*(1-(\x*0.4))))
      *((1-0.4)*1.5+0.5))
    +(1-\x)*
    (((\x*0.6)/(2.5*(1-(\x*0.4))))*0.4
    -(1-((\x*0.6)/(2.5*(1-(\x*0.4)))))*0.01)})
    --
    plot[
        domain=5/8:0.5386697094,
        samples=120,
        smooth,
        variable=\x
    ]
    ({\x},
    {\x*(0.4*1.5-0.5
    +(0.1/(0.4+0.1))*((1-0.4)*1.5+0.5))})
    -- cycle;

    \draw[
        blue,
        very thick,
        scale=5,
        domain=0:5/8,
        smooth,
        variable=\x
    ]
    plot ({\x},
    {\x*(0.4*1.5-0.5
    +(0.1/(0.4+0.1))*((1-0.4)*1.5+0.5))});

    \draw[
        blue,
        very thick,
        scale=5,
        domain=5/8:1,
        smooth,
        variable=\x
    ]
    plot ({\x},
    {\x*(0.4*1.5-0.5
    +((\x*0.6)/(2.5*(1-(\x*0.4))))
      *((1-0.4)*1.5+0.5))
    +(1-\x)*
    (((\x*0.6)/(2.5*(1-(\x*0.4))))*0.4
    -(1-((\x*0.6)/(2.5*(1-(\x*0.4)))))*0.1)});
    \node[blue] at (4.4,1.5) {$\overline{s}^t=0.1$};

    \draw[
        red,
        very thick,
        scale=5,
        domain=0:25/256,
        smooth,
        variable=\x
    ]
    plot ({\x},
    {\x*(0.4*1.5-0.5
    +(0.01/(0.4+0.01))*((1-0.4)*1.5+0.5))});

    \draw[
        red,
        very thick,
        scale=5,
        domain=25/256:1,
        smooth,
        variable=\x
    ]
    plot ({\x},
    {\x*(0.4*1.5-0.5
    +((\x*0.6)/(2.5*(1-(\x*0.4))))
      *((1-0.4)*1.5+0.5))
    +(1-\x)*
    (((\x*0.6)/(2.5*(1-(\x*0.4))))*0.4
    -(1-((\x*0.6)/(2.5*(1-(\x*0.4)))))*0.01)});
    \node[red] at (3,2.2) {$\overline{s}^t=0.01$};

    \draw[densely dotted]
        (125/256,0)--(125/256,0.0655011433);
    \draw[very thin] (125/256,0.05) -- (125/256,-0.05);
 \node[] at (3.3,-0.2) {$\mu_0^{0.1}$};
    
    \draw[densely dotted]
        (25/8,0)--(25/8,19/16);
    \draw[very thin] (25/8,0.05) -- (25/8,-0.05);
 \node[] at (.65,-0.2) {$\mu_0^{0.01}$};

\end{tikzpicture}
\caption{The defender's ex ante payoff in the audience-cost setting with $\overline{s}^t=0.1$ (blue) and $\overline{s}^t=0.01$ (red). The equilibrium is semi-separating at $\mu_0<\mu_0^{\overline s^t}$ and pooling at $\mu_0>\mu_0^{\overline s^t}$.}\label{fig.LimitedAC}
\end{center}

\end{figure}

Figure \ref{fig.LimitedAC} illustrates these effects. From Proposition \ref{AC2_Limitation}, $\mu_0^{\overline s^t}:= \frac{F^{-1}\left(\frac{\overline{s}^t}{\omega^L+\overline{s}^t}\right)}{c_c+pF^{-1}\left(\frac{\overline{s}^t}{\omega^L+\overline{s}^t}\right)}$ is the threshold on prior belief at which the equilibrium switches from semi-separating to pooling. When $\mu_0<\mu_0^{0.01}$, the equilibrium is semi-separating under both $\overline s^t=0.01$ and $0.1$, and raising $\overline  s^t$ increases the defender's payoffs. When $\mu_0>\mu_0^{0.1}$, the equilibrium is pooling under both $\overline s^t=0.01$ and $0.1$ and raising $\overline  s^t$ reduces the defender's payoffs. In between, the equilibrium is semi-separating with $\overline s^t=0.1$ and pooling with $\overline s^t=0.01$. The shaded part indicates where the payoff increases as the feasible audience cost falls, so that the equilibrium switches from semi-separating to pooling.

Having characterized the equilibria under limited audience costs, we now compare the equilibrium payoffs with those in burning-money settings. 

\begin{prop}[Small Audience Costs versus Burning Money without Commitment]\label{Comparison_BMwoComm_ACwoComm2_Limitation}
Suppose Assumption \ref{assum_2sided}-(\ref{assum_convex}), (\ref{assum_rich_F1}), and (\ref{assum_rich_s+Delta}) hold, and $\overline{s}^t<s^*$. Let $D:=(1-p)\omega^H+c_d-\omega^L$.
    \begin{itemize}
        \item If $\frac{F_1D}{\delta-F_1D}\,\omega^L<\overline{s}^t$, then the audience-cost setting yields a strictly higher defender ex ante payoff than the burning-money setting without commitment.

  \item If $\overline{s}^t<\frac{F_1D}{\delta-F_1D}\,\omega^L$, then there is a threshold $\overline \mu_0\in(0,1)$  such that the burning-money setting without commitment yields a strictly higher defender ex ante payoff than the audience-cost setting if and only if $\mu_0<\overline \mu_0$.
    \end{itemize}    
\end{prop}

Intuitively, in the pooling equilibrium of the audience-cost setting, a lower probability of the high type (smaller $\mu_0$) makes costly bluffing and backing down by the low type more frequent. At the same time, money is burned less frequently in the burning-money setting. Both effects therefore make the burning-money setting relatively more attractive as the high type becomes less likely. In the semi-separating equilibrium, as discussed above, a lower maximum feasible audience cost (smaller $\overline s^t$) reduces the defender's payoff by weakening deterrence, while leaving the burning-money setting unaffected. Once $\overline s^t$ falls below a threshold, the defender's payoff in the audience cost setting falls below that in the burning-money setting. This is the intuition for the second part of Proposition \ref{Comparison_BMwoComm_ACwoComm2_Limitation}, illustrated in Panel (b) of Figure \ref{fig_comparison}. When $\overline s^t$ is above this threshold, as in the first part of the proposition, the audience cost setting instead yields the defender a higher payoff than the burning money setting.

\begin{prop}[Small Audience Costs versus Burning Money with Commitment]\label{Comparison_BMwComm_ACwoComm2_Limitation}
   Suppose Assumption \ref{assum_2sided}-(\ref{assum_convex}), (\ref{assum_sb1_1/2}), and (\ref{assum_sb1_sbart}) hold, and $\overline{s}^t<s^*$.  
   For any $\mu_0\in(0,1)$, the defender's ex ante payoff is strictly lower in the audience-cost setting than in the burning-money setting with commitment.
\end{prop}

Intuitively, when the large audience costs are infeasible, the equilibria are not separating. Moreover, the low type bluffs and sometimes incurs the audience cost. Both effects reduce the value of the audience-cost setting to the defender. By contrast, in burning money with commitment, not only can the defender separate, but it can also do so at a minimal cost as long as a sufficiently small burning-money action is available.

\section{Probability of War}

We analyzed the equilibria and compared defender payoffs under different signaling settings. While the defender maximizes its payoffs, from a social planner's perspective, it may not internalize all the costs of war. We now compare the probability of war under different signaling settings.

\begin{prop}[Equilibrium Probability of War]\label{WarProb}
    Suppose Assumption \ref{assum_2sided}-(\ref{assum_convex}), (\ref{assum_sb1_F1}), (\ref{assum_sb1_1/2}), (\ref{assum_rich_F1}), and (\ref{assum_rich_s+Delta}) hold. In equilibrium, the burning-money setting has the same ex ante probability of war with and without commitment. Moreover,
    \begin{itemize}
        \item When $\overline{s}^t\ge s^*$, the audience-cost setting yields a strictly greater probability of war than the burning-money setting if $-p\omega^L+c_d<\frac{F_1}{1-F_1}\,\omega^L$ and $F_1<1$. If $-p\omega^L+c_d>\frac{F_1}{1-F_1}\,\omega^L$ or $F_1=1$, the two settings have the same equilibrium probability of war.         \item When $\overline{s}^t< s^*$, the audience-cost setting yields a strictly greater probability of war than the burning-money setting.
    \end{itemize}
\end{prop}

Separating equilibria minimize the probability of war because war happens only when an undeterrable challenger attacks the high type. The equilibria of the burning-money setting have this feature. By contrast, only when both large audience costs are feasible and the low type does not want to commit to war do audience-cost settings deliver the same minimal probability of war. In all other circumstances, the probability of war is strictly higher in audience-cost settings.

When large audience costs are feasible and the low type commits to fight, this increases the probability of fighting because an undeterrable challenger will attack and will be fought not only by high types but also by low types. When large audience costs are infeasible, the low type will never fight, but because separation is not possible, challengers who would not attack a high type now attack because they believe that there is a chance that the defender is a bluffing low type pretending to be the high type.

\section{Conclusion}

Academics and practitioners view signaling resolve as a key element of deterrence strategies. Two primary classes of signaling are burning money and burning bridges or creating audience costs. When the defender can create sufficiently large audience costs, it effectively commits to fighting if challenged. We proposed a new signaling approach, burning money with commitment, in which the defender maintains its flexibility to fight or back down, but ex ante commits to \textit{how} it burns money to signal resolve. We showed that this new approach can robustly improve the defender's payoffs and at the same time reduce the likelihood of war and discussed how this approach can be implemented by law and preestablished plans or automated algorithms. We also compared the payoffs of audience-cost and burning-money settings without commitment, characterizing when each dominates and highlighting the nonmonotonic relationship between the defender’s payoff and its capacity to create audience costs: the ability to create more audience costs can harm the defender.

Two directions for future research stand out. First, we focused on the canonical model of \cite{fearonSignalingForeignPolicy1997}. It will be beneficial to analyze settings with more general forms of payoffs than \eqref{eqn:u}-\eqref{eqn:v-t}. In the same vein, it may be helpful to allow for information processing costs \citep{simsImplicationsRationalInattention2003,caplinDeanRevealedPreference2015,dentiPosteriorSeparableCost2022,pomattoCostInformation2023} so that actions that are too similar to each other may be misinterpreted by adversaries. Second, it is valuable to analyze settings with imperfect commitment in which the defender is able to alter its preestablished plan ex post or intervene in the automated algorithm with some likelihood or to some extent. These directions are left to future research.

\newpage
\begin{singlespace}
\markeverypar{\the\everypar\looseness=0\relax}
\printbibliography[title={References}]
\end{singlespace}

\newpage
\appendix
\section*{Appendix}

\section{Additional Model Details}
Here, we provide a sufficient condition for Assumption \ref{assum_2sided}-(\ref{assum_convex}), which requires that the defender's interim payoff is strictly convex.
\begin{prop}\label{prop_convex}
The interim payoff $\widehat v^b(\mu,s)$ is strictly convex in $\mu\in(0,1)$ if $ \inf_{x} (\log(f(x)))^\prime\geq-2(1-p)p/c_c$, where $x\in\left(0,\frac{c_c}{1-p}\right)$. Moreover, this condition holds in the following cases.
\begin{enumerate}
    \item (Uniform) Let $v_c\sim U\left[0,\overline{v}_c\right]$ with $\overline{v}_c>\frac{c_c}{1-p}$ so that $f(x)=1/\overline{v}_c$ and $f^\prime(x)=0$ for $x\in\left(0,\frac{c_c}{1-p}\right)$. Then, the interim payoffs are convex. 
    
    \item (Exponential) Let $f(x)=\lambda e^{-\lambda x}$, where $x\geq0$ and $\lambda>0$. Then, the interim payoffs are convex if $1\leq 2(1-p)p\; (1/\lambda)/c_c= 2(1-p)p\; \mathbb{E}[v_c]/c_c$.

    \item (Lognormal) Let $f(x)=\frac{1}{x\sigma\sqrt{2\pi}}
\exp\left\{-\frac{(\log x-m)^2}{2\sigma^2}
\right\}$, where $x>0$ and $m$  is the mean of $\log v_c$. Then, the interim payoffs are convex if $1 \leq 2(1-p)p\; \sigma^2 e^{m-\sigma^2+1}/c_c=2(1-p)p\; \sigma^2 e^{(1-\sigma^2/2-\sigma^2)}\; \mathbb{E}[v_c]/c_c$.
\end{enumerate}
\end{prop}

\begin{proof}
From equation \eqref{payoff_uncertain_BM}, it suffices to examine $\widehat v^b(\mu,0)$. 
Let $\displaystyle g(\mu)=\mu c_c/(1-\mu p)$, so that $g^\prime(\mu)=c_c/(1-\mu p)^2>0$ and $g^{\prime\prime}(\mu)=2pg^\prime(\mu)/(1-\mu p)>0$. The first derivative is      \begin{eqnarray*}
        \frac{\mathrm{d}}{\mathrm{d}\mu}\widehat v^b(\mu,s=0) &=& \underbrace{((1-p)\omega^H-\omega^L+c_d)}_{>0 \text{~by~}W^L<0<W^H}\left[\underbrace{F(g(\mu))}_{\ge0}+\mu \underbrace{f(g(\mu))g^\prime(\mu)}_{\ge0}\right]\notag\\
        && +\omega^L f(g(\mu))g^\prime(\mu)+\underbrace{p\omega^H-c_d}_{>0}\\
        &>& 0.
    \end{eqnarray*}
    Next, the second derivative is
    \begin{eqnarray*}
        \frac{\mathrm{d}^2}{\mathrm{d}\mu^2}\widehat v^b(\mu,s=0) &=& \left(\mu\left((1-p)\omega^H-\omega^L+c_d\right)+\omega^L\right)\left[f^\prime(g(\mu))(g^\prime(\mu))^2+f(g(\mu))g^{\prime\prime}(\mu)\right]\notag\\
        && +2(\omega^H(1-p)-\omega^L+c_d)f(g(\mu))g^\prime(\mu).
    \end{eqnarray*}
    Rearranging this and substituting $g^\prime(\mu)=c_c/(1-\mu p)^2$ and $g^{\prime\prime}(\mu)=2c_c p/(1-\mu p)^3$, we have $ \frac{\mathrm{d}^2}{\mathrm{d}\mu^2}\widehat v^b(\mu,s=0) \ge 0$ if
    \begin{eqnarray}
        \frac{f^\prime\left(\frac{\mu c_c}{1-\mu p}\right)}{f\left(\frac{\mu c_c}{1-\mu p}\right)} &\ge& -2(1-\mu p)\frac{p}{c_c}\left[1+\frac{1-\mu p}{p}\frac{(1-p)\omega^H-\omega^L+c_d}{\mu\left((1-p)\omega^H-\omega^L+c_d\right)+\omega^L}\right]
        \label{ineq_convex}
    \end{eqnarray}
    for all $\mu$. Considering the right-hand side of inequality \eqref{ineq_convex}, we have
\begin{eqnarray*}
 -2\frac{(1-p)p}{c_c}\left[1+\frac{1- p}{p}\frac{(1-p)\omega^H-\omega^L+c_d}{(1-p)\omega^H+c_d}\right]
    &\leq& -2\frac{(1-p)p}{c_c}\left[1+\frac{(1- p)^2}{p}\frac{\omega^H-\omega^L}{\omega^H}\right]\\
    &\leq&-\frac{2(1-p)p}{c_c},
\end{eqnarray*}
where the first and second inequalities follow from $W^L<0<W^H$. Thus, inequality \eqref{ineq_convex} holds if $\min_{x} \frac{f^\prime(x)}{f(x)}
    \geq-\frac{2(1-p)p}{c_c}$.

    Now consider the case of a uniform distribution $v_c\sim U\left[0,\overline{v}_c\right]$ with $\overline{v}_c>\frac{c_c}{1-p}$. Then, for $x\in\left(0,\frac{c_c}{1-p}\right)$, we have $f(x)=1/\overline{v}_c$, $f^\prime(x)=0$, and thus $(\log f(x))^\prime=0\ge-2(1-p)p/c_c$.

    Next, consider the case of the exponential distribution. Observe $\log f(x)=\log\lambda-\lambda x$ and thus $(\log f(x))^\prime=-\lambda=\min_x(\log f(x))^\prime$. Hence, by $\mathbb{E}[v_c]=1/\lambda$, rearranging $\min_x(\log f(x))^\prime=-\lambda\ge-\frac{2(1-p)p}{c_c}$ leads to
    \begin{eqnarray*}
        1\le\frac{2(1-p)p}{c_c}\mathbb{E}[v_c].
    \end{eqnarray*}

    Finally, consider the case of the lognormal distribution. By $\log f(x)=-\log x-\log\sigma-\frac{1}{2}\log(2\pi)-\frac{(\log x-m)^2}{2\sigma^2}$, we have $(\log f(x))^\prime=-\frac{1}{x}-\frac{\log x-m}{\sigma^2x}=\frac{m-\sigma^2-\log x}{\sigma^2x}$, and thus $(\log f(x))^{\prime\prime}=\frac{\log x-m+\sigma^2-1}{\sigma^2x^2}$. Here, note that $\lim_{x\to0}(\log f(x))^\prime=+\infty$ and $\lim_{x\to+\infty}(\log f(x))^\prime=0$. Because $(\log f(x))^\prime$ is decreasing for $x<\exp\left\{m-\sigma^2+1\right\}$ and increasing for $x>\exp\left\{m-\sigma^2+1\right\}$, we obtain $\argmin (\log f(x))^\prime=\exp\left\{m-\sigma^2+1\right\}$ and $\min_x(\log f(x))^\prime=-\frac{1}{\sigma^2\exp\{m-\sigma^2+1\}}$. Thus, by $\mathbb{E}[v_c]=\exp\left\{m+\sigma^2/2\right\}$, substituting this into and rearranging $ \min_{x} (\log(f(x)))^\prime\geq-2(1-p)p/c_c$ yields
        $1\le\frac{2(1-p)p}{c_c}\sigma^2e^{1-3\sigma^2/2}\, \mathbb{E}[v_c].$
\end{proof}

Moreover, we can use the same condition for the audience-cost setting for $s<-p\omega^L+c_d$.
\begin{prop}\label{prop_convex_AC}
    $\widehat v^t(\mu,s)$ for $s<-p\omega^L+c_d$ is strictly convex in $\mu$ if $\inf_{x<c_c/(1-p)} (\log(f(x)))^\prime\geq-2(1-p)p/c_c$.
\end{prop}

\begin{proof}
By $\widehat{v}^t(\mu,0)=\widehat{v}^b(\mu,0)$ and Proposition \ref{prop_convex}, $\widehat{v}^t(\mu,0)$ is strictly convex if $\min_x(\log f(x))^\prime\geq-\frac{2(1-p)p}{c_c}$. Thus, consider the case of $s\in(0,-p\omega^L+c_d)$. Using $g(\mu)=\frac{\mu c_c}{1-\mu p}$, $g^\prime(\mu)=\frac{c_c}{(1-\mu p)^2}$, and $g^{\prime\prime}(\mu)=\frac{2c_cp}{(1-\mu p)^3}$ by Proposition \ref{prop_convex}, we can rewrite equation (\ref{payoff_uncertain_AC}) as
\begin{eqnarray*}
    \widehat v^t(\mu,s)=\mu(p\omega^H-c_d+s)-s+F(g(\mu))\left[\mu(D-s)+\omega^L+s\right],
\end{eqnarray*}
where $D=(1-p)\omega^H+c_d-\omega^L$. Differentiating this twice with respect to $\mu$ yields
\begin{eqnarray*}
    \frac{\mathrm d^2}{\mathrm d\mu^2}\widehat v^t(\mu,s)=\left[\mu(D-s)+\omega^L+s\right]\left[f^\prime(g(\mu))(g^\prime(\mu))^2+f(g(\mu))g^{\prime\prime}(\mu)\right]+2(D-s)f(g(\mu))g^\prime(\mu).
\end{eqnarray*}
The assumption of $s<-p\omega^L+c_d$ here implies
\begin{eqnarray*}
    D-s &>& D-\left(-p\omega^L+c_d\right)\\
    &=& (1-p)(\omega^H-\omega^L)>0.
\end{eqnarray*}
Then, because $f(g(\mu))>0$, $g^\prime(\mu)>0$, and $\mu(D-s)+\omega^L+s>0$, the second derivative is strictly positive if
\begin{eqnarray}
    \frac{f^\prime(g(\mu))}{f(g(\mu))}>-2(1-\mu p)\frac{p}{c_c}\left[1+\frac{1-\mu p}{p}\frac{D-s}{\mu(D-s)+\omega^L+s}\right].
\label{ineq_convex_AC1}
\end{eqnarray}
The right-hand side of inequality
(\ref{ineq_convex_AC1}) satisfies
\begin{eqnarray*}
    -2(1-\mu p)\frac{p}{c_c}\left[1+\frac{1-\mu p}{p}\frac{D-s}{\mu(D-s)+\omega^L+s} \right] &<& -2(1-\mu p)\frac{p}{c_c}\\
    &\leq& -\frac{2(1-p)p}{c_c}.
\end{eqnarray*}
Because $g(\mu)\in(0,c_c/(1-p))$ for any $\mu\in(0,1)$, $\min_{x\in(0,c_c/(1-p))}(\log f(x))^\prime\geq-\frac{2(1-p)p}{c_c}$ implies inequality (\ref{ineq_convex_AC1}). Thus, $\widehat v^t(\mu,s)$ is strictly convex in $\mu$ for every $s\in[0,-p\omega^L+c_d)$.
\end{proof}

Next, we show that the right-hand side of Assumption \ref{assum_2sided}-(\ref{assum_sb1_1/2}) is strictly positive when Assumption \ref{assum_2sided}-(\ref{assum_convex}) holds.
\begin{remark}\label{remark_v-c_positive}
 Under Assumption \ref{assum_2sided}-(\ref{assum_convex}) we have $\left(F_1-F\left(\frac{c_c}{2-p}\right)\right)\delta-F\left(\frac{c_c}{2-p}\right)\omega^L>0$.
\end{remark}
\begin{proof}
    By Assumption \ref{assum_2sided}-(\ref{assum_convex}), $\widehat{v}^t(\mu,0)=\widehat{v}^b(\mu,0)$ is strictly convex in $\mu$. Thus, $\widehat{v}^b\left(\frac{1}{2},0\right)<\frac{1}{2}\widehat{v}^b\left(0,0\right)+\frac{1}{2}\widehat{v}^b\left(1,0\right)$. Because $\widehat{v}^b\left(0,0\right)=0$ from expression (\ref{payoff_uncertain_BM}) and $F(0)=0$, this is equivalent to $2\widehat{v}^b\left(\frac{1}{2},0\right)<\widehat{v}^b\left(1,0\right)$. 
    By expression (\ref{payoff_uncertain_BM}), we obtain
    \begin{eqnarray}
    2\, \widehat{v}^b(1/2,0)&=&F\left(v^*_c(1/2)\right)\left(\omega^H+\omega^L\right)+\left(1-F\left(v^*_c(1/2)\right)\right)\left(p\omega^H-c_d\right)\label{eqn:r1}\\
\widehat{v}^b(1,0)&=&F_1\, \omega^H+\left(1-F_1\right)\left(p\omega^H-c_d\right). \label{eqn:r2}
    \end{eqnarray}
Substituting $v^*_c(1/2)=\frac{c_c}{2-p}$ into \eqref{eqn:r1} and \eqref{eqn:r2}, we have: $2\widehat{v}^b\left(\frac{1}{2},0\right)<\widehat{v}^b\left(1,0\right)$ if and only if $\left(F_1-F\left(\frac{c_c}{2-p}\right)\right)\delta-F\left(\frac{c_c}{2-p}\right)\omega^L>0$.
\end{proof}

\section{Proofs}

\begin{proof}[Proof of Proposition \ref{benchmark_BM2}]
    Consider the following cases.
    
    \textbf{Step 1} (Separating). First, consider the case in which the high type of the defender sends $s=s^\prime$ and the low type sends $s=s^{\prime\prime}$, with off-path beliefs $\mu_{s<s^\prime}=0$ and $\mu_{s>s^\prime}=1$. Observe that, in any separating equilibrium, we must have $s^{\prime\prime}=0$. Otherwise, the low type would receive a strictly negative payoff and have an incentive to deviate to $s=0$. The low type does not have the incentive to deviate to $s^\prime$ when
    \begin{eqnarray*}
        0 &\ge& F_1\omega^L+\left(1-F_1\right)\times0-s^\prime\\
        s^\prime &\ge& F_1\omega^L.
    \end{eqnarray*}
    Because of $\mu_{s<s^\prime}=0$, the low type does not deviate to any $s\in\left(0,s^\prime\right)$. Also, while $\mu_{s>s^\prime}=1$, the low type does not have the incentive to deviate to $s>s^\prime$ given $F_1\omega^L-s<F_1\omega^L-s^\prime\le0$.
    
    Under the above off-path beliefs, the high type does not have the incentive to deviate to a larger $s$ because $\mu_{s=s^\prime}$ is already $1$. Because $\mu_{s}=0$ for any $s<s^\prime$, if there is a profitable deviation from $s^\prime$, it must be $s=0$. The high type does not benefit from deviating to $s=0$ when
    \begin{eqnarray*}
        F_1\omega^H+\left(1-F_1\right)\left(p\omega^H-c_d\right)-s^\prime &\ge& p\omega^H-c_d\\
        s^\prime &\le& F_1\delta.
    \end{eqnarray*}
    Thus, $s^\prime$ must be in the interval $\left[F_1\omega^L,F_1\delta\right]$. By $W^L<0<W^H$, we have $\delta>\omega^L$, and thus, this interval is nonempty. Also, observe that the high type's payoff is decreasing in $s^\prime$. Thus, the optimal signal is $s^\prime=F_1\omega^L$, which is in $S^b$ by Assumption \ref{assum_2sided}-(\ref{assum_rich_F1}).

    For any off-path signal $s\in(0,F_1\omega^L)$, the D1 criterion assigns probability zero to the high type, and thus requires $\mu_s=0$. For $s>F_1\omega^L$, neither type can benefit even under $\mu_s=1$, and thus the criterion imposes no restriction. Thus, the off-path beliefs supporting the separating equilibrium satisfy the D1 criterion.

    \textbf{Step 2} (Pooling on $s=0$). Next, consider a pooling equilibrium where both types send $s=0$ with off-path beliefs $\mu_{s>0}\le\mu_0$. Under the pessimistic off-path beliefs, neither type has the incentive to deviate to strictly positive $s$. Thus, we examine the D1 criterion. Observe that we can express the challenger's mixed best response as the probability of not challenging, $F\left(v_c^*(\mu_s)\right)\in[0,F_1]$. Because the equilibrium payoff and the payoff from deviating to an off-path signal $s$ for the low type are $F_{\mu_0}\omega^L$ and $F\left(v^*_c(\mu_s)\right)\omega^L-s$, respectively, the low type has a weak incentive to deviate when
    \begin{eqnarray*}
        F\left(v^*_c(\mu_s)\right)\ge F_{\mu_0}+\frac{s}{\omega^L}.
    \end{eqnarray*}
    On the other hand, the high type's equilibrium payoff and the payoff from deviation are $F_{\mu_0}\omega^H+(1-F_{\mu_0})\left(p\omega^H-c_d\right)$ and $F\left(v^*_c(\mu_s)\right)\omega^H+\left(1-F\left(v^*_c(\mu_s)\right)\right)\left(p\omega^H-c_d\right)-s$, respectively. Thus, the high type has a strict incentive to deviate to an off-path $s>0$ when
    \begin{eqnarray*}
        F\left(v^*_c(\mu_s)\right)>F_{\mu_0}+\frac{s}{\delta}.
    \end{eqnarray*}
    Because $F_{\mu_0}+s/\omega^L>F_{\mu_0}+s/\delta$ by $W^L<0<W^H$, the high type has a strict incentive to deviate whenever the low type has a weak incentive to deviate. The D1 criterion requires $\mu_s=1$ for such $s$. Because this off-path signal yields $F\left(v^*_c(\mu_s)\right)=F_1$, the high type gains $\left(p\omega^H-c_d+F_1\delta-s\right)-\left(p\omega^H-c_d+F_{\mu_0}\delta\right)=\Delta \delta-s$ from deviating, where $\Delta=(F_1-F_{\mu_0})$. Thus, the high type would deviate to any feasible $s\in(0,\Delta \delta)$. By Assumption \ref{assum_2sided}-(\ref{assum_rich_s+Delta}), there exists a feasible signal in the interval. Hence, this equilibrium violates the D1 criterion.
    
    \textbf{Step 3} (Pooling on $s>0$). Consider a case in which both types of the defender send $s=s^\prime>0$ with off-path beliefs $\mu_{s\neq s^\prime}=0$. Given the off-path beliefs, the low type does not have the incentive to deviate to $s=0$ when $F_{\mu_0}\omega^L-s^\prime\ge0$, i.e., $s^\prime\le F_{\mu_0}\omega^L$. Similarly, the condition for the high type is $F_{\mu_0}\omega^H+\left(1-F_{\mu_0}\right)\left(p\omega^H-c_d\right)-s^\prime\ge p\omega^H-c_d$, i.e., $s^\prime\le F_{\mu_0}\left((1-p)\omega^H+c_d\right)$, which is always true whenever the low type does not deviate.
    For a given prior $\mu_0$, this pooling equilibrium exists when $s^\prime\le F_{\mu_0}\omega^L$ because the low type would otherwise deviate to $s=0$.
    
    We now examine the D1 criterion. The low type's equilibrium payoff and the payoff from deviating to some $s>s^\prime$ are $F_{\mu_0}\omega^L-s^\prime$ and $F\left(v^*_c(\mu_s)\right)\omega^L-s$, respectively. Hence, the low type is weakly incentivized to deviate to such $s$ when
    \begin{eqnarray*}
        F\left(v^*_c(\mu_s)\right)\ge F_{\mu_0}+\frac{s-s^\prime}{\omega^L}.
    \end{eqnarray*}
    The high type's equilibrium payoff and the payoff from deviating to $s>s^\prime$ are, respectively, $F_{\mu_0}\omega^H+(1-F_{\mu_0})\left(p\omega^H-c_d\right)-s^\prime$ and $F\left(v^*_c(\mu_s)\right)\omega^H+\left(1-F\left(v^*_c(\mu_s)\right)\right)\left(p\omega^H-c_d\right)-s$. Thus, a deviation to $s>s^\prime$ is strictly beneficial for the high type when
    \begin{eqnarray*}
        F\left(v^*_c(\mu_s)\right)>F_{\mu_0}+\frac{s-s^\prime}{\delta}.
    \end{eqnarray*}
    By $W^L<0<W^H$, we know that $F_{\mu_0}+\left(s-s^\prime\right)/\omega^L>F_{\mu_0}+\left(s-s^\prime\right)/\delta$, and thus, the high type has a strict incentive to deviate to $s>s^\prime$ whenever the low type has a weak incentive to do so. The D1 criterion assigns $\mu_s=1$ to such $s$. Because we have $F\left(v^*_c(\mu_s)\right)=F_1$ for such $s$, the high type's benefit from deviating from $s^\prime>0$ to $s>s^\prime$ is $\Delta \delta-s+s^\prime$. Therefore, the high type has a strictly profitable deviation to any feasible $s\in\left(s^\prime,s^\prime+\Delta \delta\right)$. Assumption \ref{assum_2sided}-(\ref{assum_rich_s+Delta}) guarantees that the interval has a feasible signal for any $s^\prime\in\left(0,F_{\mu_0}\omega^L\right]$. Thus, a pooling equilibrium with $s^\prime>0$ violates the D1 criterion.

    \textbf{Step 4} (Semi-separating). Consider a case in which the high type sends $s=s^\prime>0$, and the low type sends $s=s^\prime$ with probability $q\in(0,1)$ and $s=0$ with probability $1-q$ with off-path beliefs $\mu_{s<s^\prime}=0$ and $\mu_{s>s^\prime}=1$. It is easy to see that the signal only from the low type must be $s=0$ because the posterior is zero; for any positive $s$, the low type would deviate to a smaller $s$. We show that a semi-separating equilibrium in which the low type mixes violates the D1 criterion.
    
    Because the low type's indifference condition between sending $s=s^\prime$ and $s=0$ is $F\left(v^*_c(\mu_{s=s^\prime})\right)\omega^L-s^\prime=0$, we obtain $s^\prime=F\left(v^*_c(\mu_{s=s^\prime})\right)\omega^L$. Also observe that, by Bayes' rule, the posterior $\mu_{s=s^\prime}=\frac{\mu_0}{\mu_0+(1-\mu_0)q}\in(\mu_0,1)$ implies $s^\prime=F\left(v^*_c(\mu_{s=s^\prime})\right)\omega^L\in(F_{\mu_0}\omega^L,F_1\omega^L)$. Consider an off-path signal $s=F_1\omega^L$, which is feasible by Assumption \ref{assum_2sided}-(\ref{assum_rich_F1}). The low type has a weakly profitable deviation to $s=F_1\omega^L$ when
    \begin{eqnarray*}
        F\left(v^*_c(\mu_{s=F_1\omega^L})\right)\ge F\left(v^*_c(\mu_{s=s^\prime})\right)+\frac{F_1\omega^L-s^\prime}{\omega^L}
    \end{eqnarray*}
    because the equilibrium payoff and the payoff from deviation are $F\left(v^*_c(\mu_{s=s^\prime})\right)\omega^L-s^\prime$ and $F\left(v^*_c(\mu_{s=F_1\omega^L})\right)\omega^L-F_1\omega^L$, respectively. The high type's equilibrium payoff and the payoff from deviation are $F\left(v^*_c(\mu_{s=s^\prime})\right)\omega^H+\left(1-F\left(v^*_c(\mu_{s=s^\prime})\right)\right)\left(p\omega^H-c_d\right)-s^\prime$ and $F\left(v^*_c(\mu_{s=F_1\omega^L})\right)\omega^H+\left(1-F\left(v^*_c(\mu_{s=F_1\omega^L})\right)\right)\left(p\omega^H-c_d\right)-F_1\omega^L$. Hence, the high type has a strictly profitable deviation to $s=F_1\omega^L$ when
    \begin{eqnarray*}
        F\left(v^*_c(\mu_{s=F_1\omega^L})\right)> F\left(v^*_c(\mu_{s=s^\prime})\right)+\frac{F_1\omega^L-s^\prime}{\delta}.
    \end{eqnarray*}
    Because we know $s^\prime<F_1\omega^L$ and $\delta>\omega^L$ by $W^L<0<W^H$, we have $F\left(v^*_c(\mu_{s=s^\prime})\right)+\frac{F_1\omega^L-s^\prime}{\omega^L}>F\left(v^*_c(\mu_{s=s^\prime})\right)+\frac{F_1\omega^L-s^\prime}{\delta}$. Thus, the high type has a strictly profitable deviation to $s=F_1\omega^L$ whenever the low type has a weakly profitable deviation. Thus, by the D1 criterion, we obtain $\mu_{s=F_1\omega^L}=1$. Then, the high type's profit from deviating to $s=F_1\omega^L$ is
    \begin{eqnarray*}
        \left(F_1-F\left(v^*_c(\mu_{s=s^\prime})\right)\right)\delta-F_1\omega^L+s^\prime=\left(F_1-F\left(v^*_c(\mu_{s=s^\prime})\right)\right)\left(\delta-\omega^L\right)>0.
    \end{eqnarray*}
    Because the high type has an incentive to deviate, a semi-separating equilibrium in which the low type mixes violates the D1 criterion.

    Next, consider a semi-separating equilibrium in which the low type sends $s=s^\prime$, and the high type sends $s=s^\prime$ with probability $q\in(0,1)$ and $s=s^{\prime\prime}$ with probability $1-q$, with off-path beliefs $\mu_{s<s^\prime}=0$, $\mu_{s\in\left(s^\prime,s^{\prime\prime}\right)}=\mu_{s=s^\prime}$, and $\mu_{s>s^{\prime\prime}}=1$. We show that a semi-separating equilibrium in which the high type mixes violates the D1 criterion.
    
    By Bayes' rule, we have $\mu_{s=s^\prime}=\frac{\mu_0q}{\mu_0q+(1-\mu_0)}<\mu_0$ and thus $F\left(v^*_c(\mu_{s=s^\prime})\right)<F_{\mu_0}$. The indifference condition of the high type yields
    \begin{eqnarray*}
        F\left(v^*_c(\mu_{s=s^\prime})\right)\omega^H+\left(1-F\left(v^*_c(\mu_{s=s^\prime})\right)\right)\left(p\omega^H-c_d\right)-s^\prime &=& F_1\omega^H+(1-F_1)\left(p\omega^H-c_d\right)-s^{\prime\prime}\\
        s^{\prime\prime} &=& s^\prime+\left(F_1-F\left(v^*_c(\mu_{s=s^\prime})\right)\right)\delta.
    \end{eqnarray*}
    We consider an off-path signal $s\in\left(s^\prime,s^{\prime\prime}\right)$. Assumption \ref{assum_2sided}-(\ref{assum_rich_s+Delta}) guarantees that there exists a feasible signal in the interval $\left(s^\prime,s^{\prime\prime}\right)$. To see this, observe that $s^\prime\le F\left(v^*_c(\mu_{s=s^\prime})\right)\omega^L<F_{\mu_0}\omega^L$ by $F\left(v^*_c(\mu_{s=s^\prime})\right)\omega^L-s^\prime\ge0$ and $F\left(v^*_c(\mu_{s=s^\prime})\right)<F_{\mu_0}$. Also observe that $s^{\prime\prime}=s^\prime+\left(F_1-F\left(v^*_c(\mu_{s=s^\prime})\right)\right)\delta>s^\prime+\Delta \delta$ by $F\left(v^*_c(\mu_{s=s^\prime})\right)<F_{\mu_0}$. Thus, we have $\left(s^\prime,s^\prime+\Delta \delta\right)\subset\left(s^\prime,s^{\prime\prime}\right)$. Then, there is some $s$ in $\left(s^\prime,s^{\prime\prime}\right)$ by Assumption \ref{assum_2sided}-(\ref{assum_rich_s+Delta}).
    
    The low type's equilibrium payoff and the payoff from deviating to $s\in\left(s^\prime,s^{\prime\prime}\right)$ are $F\left(v^*_c(\mu_{s=s^\prime})\right)\omega^L-s^\prime$ and $F\left(v^*_c(\mu_{s\in\left(s^\prime,s^{\prime\prime}\right)})\right)\omega^L-s$, respectively. Thus, the low type has a weak incentive to deviate to such $s$ when
    \begin{eqnarray*}
        F\left(v^*_c(\mu_{s\in\left(s^\prime,s^{\prime\prime}\right)})\right)\ge F\left(v^*_c(\mu_{s=s^\prime})\right)+\frac{s-s^\prime}{\omega^L}.
    \end{eqnarray*}
    On the other hand, the high type's equilibrium payoff and the payoff from deviation are respectively
    \begin{eqnarray*}
        && F\left(v^*_c(\mu_{s=s^\prime})\right)\omega^H+\left(1-F\left(v^*_c(\mu_{s=s^\prime})\right)\right)\left(p\omega^H-c_d\right)-s^\prime~\text{and}\\
        && F\left(v^*_c(\mu_{s\in\left(s^\prime,s^{\prime\prime}\right)})\right)\omega^H+\left(1-F\left(v^*_c(\mu_{s\in\left(s^\prime,s^{\prime\prime}\right)})\right)\right)\left(p\omega^H-c_d\right)-s.
    \end{eqnarray*}
    Hence, the high type has a strictly profitable deviation to $s\in\left(s^\prime,s^{\prime\prime}\right)$ when
    \begin{eqnarray*}
        F\left(v^*_c(\mu_{s\in\left(s^\prime,s^{\prime\prime}\right)})\right)>F\left(v^*_c(\mu_{s=s^\prime})\right)+\frac{s-s^\prime}{\delta}.
    \end{eqnarray*}
    By $\delta>\omega^L$ and $s>s^\prime$, we have $F\left(v^*_c(\mu_{s=s^\prime})\right)+\frac{s-s^\prime}{\omega^L}>F\left(v^*_c(\mu_{s=s^\prime})\right)+\frac{s-s^\prime}{\delta}$. Hence, whenever the low type has a weakly profitable deviation to $s\in\left(s^\prime,s^{\prime\prime}\right)$, the high type has a strictly profitable deviation to such $s$. This yields $\mu_{s\in\left(s^\prime,s^{\prime\prime}\right)}=1$. By deviating to $s\in\left(s^\prime,s^{\prime\prime}\right)$, the high type gains
    \begin{eqnarray*}
        \left(F_1-F\left(v^*_c(\mu_{s=s^\prime})\right)\right)\delta-\left(s-s^\prime\right).
    \end{eqnarray*}
    This is strictly positive by $s<s^{\prime\prime}=s^\prime+\left(F_1-F\left(v^*_c(\mu_{s=s^\prime})\right)\right)\delta$. Thus, because the high type deviates to $s\in\left(s^\prime,s^{\prime\prime}\right)$, a semi-separating equilibrium in which the high type mixes violates the D1 criterion.

    Finally, consider a case where both types send $s=s^\prime$ and $s=s^{\prime\prime}$ with positive probability, where $s^\prime<s^{\prime\prime}$. Suppose that such an equilibrium exists. Then, the indifference condition of the low type requires $F\left(v^*_c(\mu_{s=s^\prime})\right)\omega^L-s^\prime=F\left(v^*_c(\mu_{s=s^{\prime\prime}})\right)\omega^L-s^{\prime\prime}$, and thus 
    \begin{eqnarray}\label{eqn:indiff_L}
        s^{\prime\prime}-s^\prime=\left(F\left(v^*_c(\mu_{s=s^{\prime\prime}})\right)-F\left(v^*_c(\mu_{s=s^\prime})\right)\right)\omega^L.
    \end{eqnarray}
    Next, the high type's indifference condition yields
    {\small \begin{eqnarray}\label{eqn:indiff_H}
        F\left(v^*_c(\mu_{s=s^\prime})\right)\omega^H+\left(1-F\left(v^*_c(\mu_{s=s^\prime})\right)\right)W^H-s^\prime = F\left(v^*_c(\mu_{s=s^{\prime\prime}})\right)\omega^H+\left(1-F\left(v^*_c(\mu_{s=s^{\prime\prime}})\right)\right)W^H-s^{\prime\prime}\nonumber
         \end{eqnarray} }\vspace{-4em}
         \begin{eqnarray}
        s^{\prime\prime}-s^\prime =\left(F\left(v^*_c(\mu_{s=s^{\prime\prime}})\right)-F\left(v^*_c(\mu_{s=s^\prime})\right)\right)\delta.
    \end{eqnarray}
    By $\delta>\omega^L$ and equalities \eqref{eqn:indiff_L} and \eqref{eqn:indiff_H}, the two indifference conditions cannot hold simultaneously, a contradiction.

    \textbf{Step 5} (defender-optimal equilibrium). Because only the separating equilibrium in Step 1 satisfies our equilibrium refinement, the D1 criterion, it is the defender-optimal equilibrium.
\end{proof}

\begin{proof}[Proof of Proposition \ref{Prop_BM_2sided}]
    We first specify the signal structure in the case with commitment power. Then, we show that the defender's ex ante payoff under this signal structure is the highest possible payoff in the setting, and thus commitment power is strictly beneficial for the defender.

    \textbf{Step 1} (Signal structure). First, we specify the signal structure in the case with commitment power and the ex ante payoff of the defender. When $\mu_0$ is small, we have $\mu_{s=0}=0,\mu_{s=s^b_1}=1$. Bayes-Plausibility requires $\tau(\mu_{s=0})\mu_{s=0}+(1-\tau(\mu_{s=0}))\mu_{s=s^b_1}=\mu_0$, hence
    \begin{eqnarray*}
        \tau(\mu_{s=0})=1-\mu_0~\text{and}~\tau(\mu_{s=s^b_1})=\mu_0.
    \end{eqnarray*}
    The defender's signal structure follows from the above values:
    \begin{eqnarray}\label{eqn:pi_small_mu}
        \pi(0|\omega^H)=0,~\pi(s=s^b_1|\omega^H)=1,~\pi(0|\omega^L)=1,~\text{and}~\pi(s=s^b_1|\omega^L)=0.
    \end{eqnarray}
    When $\mu_0$ is large, we have $\mu_{s=0}=1$ and $\mu_{s=s^b_1}=0$.
    Bayes-Plausibility yields
    \begin{eqnarray*}
        \tau(\mu_{s=0})=\mu_0~\text{and}~\tau(\mu_{s=s^b_1})=1-\mu_0.
    \end{eqnarray*}
    Hence, we obtain
    \begin{eqnarray}\label{eqn:pi_large_mu}
        \pi(0|\omega^H)=1,~\pi(s^b_1|\omega^H)=0,~\pi(0|\omega^L)=0,~\text{and}~\pi(s^b_1|\omega^L)=1.
    \end{eqnarray}

    By \eqref{eqn:pi_small_mu}, the ex ante payoff when $\mu_0$ is small is
    \begin{eqnarray*}
        \mu_0\left[F_1\omega^H+(1-F_1)W^H-s^b_1\right]+(1-\mu_0)\times0=\mu_0\left(\widehat{v}^b(1,0)-s^b_1\right),
    \end{eqnarray*}
    whereas the payoff when $\mu_0$ is large is
    \begin{eqnarray*}
        \mu_0\left[F_1\omega^H+(1-F_1)W^H\right]-(1-\mu_0)s^b_1=\mu_0\widehat{v}^b(1,0)-(1-\mu_0)s^b_1.
    \end{eqnarray*}
    The above payoffs become equal at $\mu_0=1/2$, and thus, the ex ante defender payoff in the burning-money setting with commitment is expressed as
    \begin{eqnarray}\label{eqn:comm_payoff}
        \mu_0\widehat{v}^b(1,0)-\min\left\{\mu_0,1-\mu_0\right\}s^b_1=V^{jo}\left(\mu_0\middle|(\widehat v_s^b)_{s\in S^b}\right).
    \end{eqnarray}

    \textbf{Step 2} (Payoff optimality). We want to show that $V^{jo}\left(\mu_0\middle|(\widehat v_s^b)_{s\in S^b}\right)$ in \eqref{eqn:comm_payoff} is the highest possible payoff the defender can achieve.     That is, we want to show
    \begin{eqnarray*}
        \sum_{s\in S^b}\tau(\mu_s)\widehat{v}^b(\mu_s,s)\le V^{jo}\left(\mu_0\middle|(\widehat v_s^b)_{s\in S^b}\right)
    \end{eqnarray*}
    for any Bayes-plausible signal structure. First, consider any $s>0$ and $\mu\in[0,1]$. Then, by Assumption \ref{assum_2sided}-(\ref{assum_convex}), we obtain
    \begin{eqnarray*}
        \widehat v^b(\mu,s>0)=\widehat v^b(\mu,0)-s\le \mu\widehat v^b(1,0)+(1-\mu)\underbrace{\widehat v^b(0,0)}_{=0}-s^b_1.
    \end{eqnarray*}

    Next, consider $s=0$ and $\mu_{s=0}\le1/2$. First, by Assumption \ref{assum_2sided}-(\ref{assum_sb1_1/2}), observe
    \begin{eqnarray}\label{eqn:convex_1/2}
        2\widehat v^b\left(\frac{1}{2},0\right) &=& 2\times\left(\frac{1}{2}\left(W^H+F_{0.5}\delta\right)+\frac{1}{2}F_{0.5}\omega^L\right)\nonumber\\
        &=& W^H+F_1\delta-\left[(F_1-F_{0.5})\delta-F_{0.5}\omega^L\right]\nonumber\\
        &<& \widehat v^b(1,0)-s^b_1.
    \end{eqnarray}
    Thus, for any $\mu_{s=0}\le1/2$, we have
    \begin{eqnarray*}
        \widehat v^b(\mu_{s=0},0)\le 2\mu_{s=0}\widehat v^b\left(\frac{1}{2},0\right)+(1-2\mu_{s=0})\widehat v^b(0,0)\le\mu_{s=0}\left(\widehat v^b(1,0)-s^b_1\right),
    \end{eqnarray*}
    by Assumption \ref{assum_2sided}-(\ref{assum_convex}) and $\widehat v^b(0,0)=0$. Then, by Bayes-plausibility, we obtain the ex ante payoff
    \begin{eqnarray}\label{eqn:exante_small}
        \sum_{s\in S^b}\tau(\mu_s)\widehat v^b(\mu_s,s) &\le& \sum_{s\in S^b}\tau(\mu_s)\mu_s\left(\widehat v^b(1,0)-s^b_1\right)\nonumber\\
        &=& \mu_0\left(\widehat v^b(1,0)-s^b_1\right).
    \end{eqnarray}

    Lastly, consider $s=0$ and $\mu_{s=0}\ge1/2$. By Assumption \ref{assum_2sided}-(\ref{assum_convex}) and inequality \eqref{eqn:convex_1/2}, we have
    \begin{eqnarray*}
        \widehat v^b(\mu_{s=0},0) &\le& 2(1-\mu_{s=0})\widehat v^b\left(\frac{1}{2},0\right)+(2\mu_{s=0}-1)\widehat v^b(1,0)\\
        &\le& (1-\mu_{s=0})\left(\widehat v^b(1,0)-s^b_1\right)+(2\mu_{s=0}-1)\widehat v^b(1,0)\\
        &=& \mu_{s=0}\widehat v^b(1,0)-(1-\mu_{s=0})s^b_1.
    \end{eqnarray*}
    Then, by Bayes-plausibility, we obtain the ex ante payoff
    \begin{eqnarray}\label{eqn:exante_large}
        \sum_{s\in S^b}\tau(\mu_s)\widehat v^b(\mu_s,s) &\le& \sum_{s\in S^b}\tau(\mu_s)\left(\mu_s\widehat v^b(1,0)-(1-\mu_s)s^b_1\right)\nonumber\\
        &=& \mu_0\widehat v^b(1,0)-(1-\mu_0)s^b_1.
    \end{eqnarray}
    In sum, every Bayes-plausible signal structure satisfies at least one of \eqref{eqn:exante_small} and \eqref{eqn:exante_large}. Hence, we have
    \begin{eqnarray*}
        \sum_{s\in S^b}\tau(\mu_s)\widehat v^b(\mu_s,s) &\le& \max\left\{\mu_0\left(\widehat v^b(1,0)-s^b_1\right),\mu_0\widehat v^b(1,0)-(1-\mu_0)s^b_1\right\}\\
        &=& V^{jo}\left(\mu_0\middle|(\widehat v_s^b)_{s\in S^b}\right)
    \end{eqnarray*}
    for any Bayes-plausible signal structure.

    Finally, the above signal structures strictly improve the defender's ex ante payoff, compared to the case without commitment. We already know that the ex ante payoff is bounded by $V^{jo}\left(\mu_0\middle|(\widehat v_s^b)_{s\in S^b}\right)$. Then, we show that the case without commitment power cannot attain it and, thus, yields a strictly smaller payoff. First, suppose $\mu_0\le1/2$. The defender obtains the ex ante payoff $\mu_0\left(\widehat v^b(1,0)-s^b_1\right)$ by sending $s=s^b_1$ when $\omega=\omega^H$ and sending $s=0$ when $\omega=\omega^L$. If the defender does not have commitment power, the low type can profitably deviate from $s=0$ to $s^b_1$ because $F_1\omega^L-s^b_1>0$ by Assumption \ref{assum_2sided}-(\ref{assum_sb1_F1}). Second, suppose $\mu_0\ge1/2$. The defender achieves the ex ante payoff $\mu_0\widehat v^b(1,0)-(1-\mu_0)s^b_1$ by sending $s=0$ when $\omega=\omega^H$ and sending $s=s^b_1$ when $\omega=\omega^L$. If the defender does not have commitment power, then the low type can profitably deviate from $s=s^b_1>0$ to $s=0$. In sum, the signal structures in \eqref{eqn:pi_small_mu} and \eqref{eqn:pi_large_mu} cannot be part of an equilibrium in the burning-money setting without commitment. Thus, the case without commitment power does not yield the same ex ante payoff as the setting with commitment.
\end{proof}

\begin{proof}[Proof of Proposition \ref{benchmark_AC2_BridgeBurning}]
    We first consider a special case of $F_1=1$. Because we have adopted the interpretation that $1/(1-F(x))=\infty$ when $F(x)=1$, we have $-p\omega^L+c_d<\frac{F_1}{1-F_1}\,\omega^L=\infty$. Because this implies that the challenger never challenges if the defender fights for sure, both types of the defender send some $s\ge-p\omega^L+c_d$. Consequently, the high and low types receive $\omega^H$ and $\omega^L$, and they do not have the incentive to deviate to any $s<-p\omega^L+c_d$. Moreover, because $\omega^H$ and $\omega^L$ are the defender's highest possible payoffs, this equilibrium satisfies the D1 criterion.
    
    Henceforth, we focus on the case of $F_1<1$ and consider the following cases.
    
    \textbf{Case 1} ($-p\omega^L+c_d\ge\frac{F_1}{1-F_1}\omega^L$). Consider a separating equilibrium where the high-valuation defender sends $s=s^\prime$ and the low type sends $s=s^{\prime\prime}$. As shown in Step 1 of the proof of Proposition \ref{benchmark_BM2}, we have $s^{\prime\prime}=0$ in any separating equilibrium. Because $\mu_{s=s^\prime}=1$ and $\mu_{s=0}=0$, the low type does not have the incentive to deviate to $s^\prime>0$ as long as
    \begin{eqnarray*}
        \underbrace{0}_{\text{Equilibrium payoff}} &\ge& \underbrace{F_1\omega^L-\left(1-F_1\right)s^\prime}_{\text{Deviation to }s=s^\prime}\\
        s^\prime &\ge& \frac{F_1}{1-F_1}\omega^L.
    \end{eqnarray*}
    Observe that the equilibrium payoff for the low type is zero because, by $F(0)=0$, the challenger challenges for sure when $\mu_{s=0}=0$. When $-p\omega^L+c_d\ge\frac{F_1}{1-F_1}\omega^L$, the low type does not have the incentive to deviate to $s\ge-p\omega^L+c_d$ because the inequality is equivalent to $0\ge F_1\omega^L+(1-F_1)\left(p\omega^L-c_d\right)$. Next, the high type's payoff from sending $s=s^\prime$ is $F_1\omega^H+\left(1-F_1\right)\left(p\omega^H-c_d\right)$, which does not depend on $s^\prime$. Thus, any feasible $s^\prime\ge\frac{F_1}{1-F_1}\omega^L$ is optimal for the high type. Our assumption $\overline{s}^t\ge s^*=\min\left\{ -p\omega^L+c_d, \frac{F_1}{1-F_1}\, \omega^L \right\}$ guarantees that such $s^\prime$ is feasible when $-p\omega^L+c_d\ge\frac{F_1}{1-F_1}\omega^L$. Off-path beliefs $\mu_{s<\frac{F_1}{1-F_1}\omega^L}=0$ and $\mu_{s\ge\frac{F_1}{1-F_1}\omega^L}=1$ support this equilibrium.

    This equilibrium satisfies the D1 criterion. First, consider an off-path signal $s<\frac{F_1}{1-F_1}\,\omega^L$. The high type's payoff from deviation under the most favorable posterior is $F_1\omega^H+\left(1-F_1\right)\left(p\omega^H-c_d\right)$, which is equal to his equilibrium payoff. On the other hand, the low type would receive $F_1\omega^L-(1-F_1)s>0$ under the most favorable posterior. Thus, whenever the high type has a weakly profitable deviation, the low type has a strictly profitable deviation. Thus, the D1 criterion assigns $\mu_s=0$ for such an off-path signal. The challenger would always challenge after $s<\frac{F_1}{1-F_1}\,\omega^L$, lowering both types' payoffs. Second, consider an off-path signal $s\in\left(\frac{F_1}{1-F_1}\,\omega^L,-p\omega^L+c_d\right)$. As in the first case, this deviation does not improve the high type's payoff. For the low type, a deviation to $s\in\left(\frac{F_1}{1-F_1}\,\omega^L,-p\omega^L+c_d\right)$ would yield at most $F_1\omega^L-(1-F_1)s<0$. Hence, neither type can have a strict incentive to deviate. Third, consider an off-path signal $s\ge-p\omega^L+c_d\ge\frac{F_1}{1-F_1}\,\omega^L$. This deviation never strictly improves the defender's payoff. In sum, the separating equilibrium satisfies the D1 criterion.

    We show that, when $c_d-p\omega^L>\frac{F_1}{1-F_1}\omega^L$, any equilibrium satisfying the D1 criterion must take the form of the above equilibrium, (a) which is a separating equilibrium, (b) in which the low type sends $s=0$, and (c) in which the high type sends some $s\ge\frac{F_1}{1-F_1}\,\omega^L$. First, suppose to the contrary that feature (a) does not hold: consider a semi-separating or a pooling equilibrium. This implies that some common $s$ is sent by both types with positive probability. Observe that this common signal must satisfy $s<\frac{F_1}{1-F_1}\,\omega^L$. The low type never sends $s>\frac{F_1}{1-F_1}\,\omega^L$. Also, the low type does not send $s=\frac{F_1}{1-F_1}\,\omega^L$, either, because a common signal induces $\mu_s<1$, which gives the low type a strictly negative payoff. Because the low type does not send $s\ge\frac{F_1}{1-F_1}\,\omega^L$, a common signal must be smaller than $\frac{F_1}{1-F_1}\,\omega^L$. Next, fix some feasible $s^\prime\ge\frac{F_1}{1-F_1}\,\omega^L$. Because the low type never sends such $s^\prime$, we have $\mu_{s=s^\prime}=1$ if the high type sends it by Bayes' rule, and the high type would achieve his highest possible payoff. However, an equilibrium that is not separating implies that some common $s$ is sent by both types with positive probability. This common signal leads to $\mu_s<1$, implying that the high type's payoff becomes lower than the highest possible one. Thus, the high type cannot send $s^\prime$ by the indifference condition. Thus, $s^\prime$ must be an off-path signal when feature (a) does not hold. From this large off-path signal $s^\prime$, the low type never benefits even under the most favorable posterior, $\mu_{s=s^\prime}=1$. On the other hand, the high type would strictly benefit from such $s^\prime$. Then, the D1 criterion imposes $\mu_{s=s^\prime}=1$ where $s^\prime\ge\frac{F_1}{1-F_1}\,\omega^L$, and the high type deviates to the large signal, a contradiction.
    
    Second, suppose that feature (b) does not hold: consider a separating equilibrium where the low type sends $s>0$. If $s<-p\omega^L+c_d$, he would receive $-s<0$. If $s\ge-p\omega^L+c_d$, he would receive $F_1\omega^L+(1-F_1)\left(p\omega^L-c_d\right)<0$. Thus, this case cannot constitute an equilibrium because the low type has the incentive to deviate to $s=0$.
    
    Third, suppose that feature (c) does not hold: consider a separating equilibrium where the high type sends $s<\frac{F_1}{1-F_1}\,\omega^L$. Then, because $\mu_s=1$, the low type has the incentive to imitate and receive $F_1\omega^L-(1-F_1)s>0$. In sum, any equilibrium must take the form of the above separating equilibrium, implying that it is defender-optimal.

    \textbf{Case 2} ($-p\omega^L+c_d\le\frac{F_1}{1-F_1}\omega^L$). Next, consider a case in which both types of the defender send some $s\ge-p\omega^L+c_d$ with off-path beliefs $\mu_{s<-p\omega^L+c_d}=0$, where both types are willing to fight on the equilibrium path. By sending some $s\ge c_d-p\omega^L$, the high type obtains the highest possible payoff $F_1\omega^H+\left(1-F_1\right)\left(p\omega^H-c_d\right)$ because the high type never pays the cost of the signal. Now the low type is also willing to fight and does not pay audience costs. Sending $s\ge c_d-p\omega^L$ generates a weakly higher payoff for the low type than sending $s=0$ and then acquiescing when
    \begin{eqnarray*}\label{eq_benchmark_AC2_BridgeBurning1}
        F_1\omega^L+\left(1-F_1\right)\left(p\omega^L-c_d\right) &\ge& 0\\
        -p\omega^L+c_d\le\frac{F_1}{1-F_1}\,\omega^L,\nonumber
    \end{eqnarray*}
    which is our assumption here.

    This equilibrium satisfies the D1 criterion. First, any off-path signal such that $s\ge-p\omega^L+c_d$ does not change the defender's payoff. Second, consider an off-path signal $s<-p\omega^L+c_d$. The high type would receive at most his equilibrium payoff from this deviation. Under the most favorable belief, the low type would receive $F_1\omega^L-(1-F_1)s>0$. The benefit for the low type from this deviation is
    \begin{eqnarray*}
        \left(F_1\omega^L-(1-F_1)s\right)-\left(F_1\omega^L+(1-F_1)\left(p\omega^L-c_d\right)\right)=(1-F_1)\left(c_d-p\omega^L-s\right)>0.
    \end{eqnarray*}
    Thus, whenever the high type weakly prefers deviation to $s<-p\omega^L+c_d$, the low type has a strict incentive to deviate to such $s$. The D1 criterion imposes $\mu_s=0$. This implies that the high type and the low type would receive $p\omega^H-c_d$ and $-s$, respectively, by deviating to $s<-p\omega^L+c_d$, both of which are strictly smaller than their equilibrium payoffs. Thus, this equilibrium satisfies the D1 criterion. 

    There is no other equilibrium with a higher defender payoff. In any equilibrium, the high type must achieve the highest possible payoff, $F_1\omega^H+(1-F_1)\left(p\omega^H-c_d\right)$, because he can always commit to fighting and receive this payoff. First, the two types of the defender never send a common $s<-p\omega^L+c_d$ with positive probability. This is because such pooling leads to $\mu_s<1$, which lowers the high type's payoff. Second, given the first argument, the low type sending $s<-p\omega^L+c_d$ cannot be part of an equilibrium because it would yield $\mu_s=0$ and the low type would receive $-s\le0$, which is strictly smaller than the payoff he can guarantee by committing to fighting. Thus, the low type sends only $s\ge-p\omega^L+c_d$. Lastly, consider any case where the high type sends $s<-p\omega^L+c_d$ with positive probability and the low type sends only $s\ge-p\omega^L+c_d$. Such an equilibrium does not exist because the low type would profitably deviate to the high type's signal, as $F_1\omega^L-(1-F_1)s>F_1\omega^L+(1-F_1)(p\omega^L-c_d)$ whenever $s<-p\omega^L+c_d$. Thus, the remaining candidate is the above case in which both types send only $s\ge-p\omega^L+c_d$.

    \textbf{Case 3} ($-p\omega^L+c_d=\frac{F_1}{1-F_1}\omega^L$). When $-p\omega^L+c_d=\frac{F_1}{1-F_1}\omega^L$, both types of the above equilibria are defender-optimal. Moreover, in this knife-edge case, there is another class of equilibria. Consider a case in which the high type sends some $s^\prime\ge-p\omega^L+c_d$ and the low type sends $s=s^\prime$ with probability $q$ and $s=0$ with probability $1-q$. While $\mu_{s=s^\prime}\in[0,1]$ should follow from Bayes' rule, we know that both types commit to fighting by sending $s^\prime\ge-p\omega^L+c_d$. This implies that the challenger with $v_c>\frac{c_c}{1-p}$ challenges and that the probability of challenge is $1-F_1$. Consequently, the low type's payoff from sending $s=s^\prime$ when $-p\omega^L+c_d\le\frac{F_1}{1-F_1}\omega^L$ is
    \begin{eqnarray*}
        F_1\omega^L+(1-F_1)\left(p\omega^L-c_d\right)=F_1\omega^L-(1-F_1)\frac{F_1}{1-F_1}\,\omega^L=0.
    \end{eqnarray*}
    Because the payoff from sending $s=0$ is zero, we have established the low type's indifference condition. The low type may use an arbitrary $q\in[0,1]$. Similarly, by committing to fight, the high type obtains the highest possible payoff, $F_1\omega^H+(1-F_1)\left(p\omega^H-c_d\right)$, regardless of $q$.

    This class of equilibrium satisfies the D1 criterion. Consider an off-path signal $s\in\left(0,-p\omega^L+c_d\right)$. First, the high type cannot strictly improve his payoff. Second, because the low type would acquiesce to a challenge after sending $s\in\left(0,-p\omega^L+c_d\right)$ and because his equilibrium payoff is zero, he has a strict incentive to deviate to $s\in\left(0,-p\omega^L+c_d\right)$ when
    \begin{eqnarray*}
        F\left(v^*_c(\mu_s)\right)>\frac{s}{\omega^L+s}.
    \end{eqnarray*}
    Thus, whenever the high type has a weak incentive to deviate to $s\in\left(0,-p\omega^L+c_d\right)$ (which occurs only when $F\left(v^*_c(\mu_s)\right)=F_1$), the low type has a strict incentive to do so. The D1 criterion thus assigns $\mu_s=0$ for such $s$. With $\mu_s=0$, both types would receive payoffs strictly smaller than the equilibrium payoff. Therefore, neither type deviates, and thus this class of equilibria satisfies the D1 criterion. 
\end{proof}

\begin{proof}[Proof of Proposition \ref{Comparison_BMwoComm_ACwoComm2}]
    We compare burning money with commitment and audience cost. By Proposition \ref{AC2_Large}, the defender's ex ante payoff in the audience-cost setting is at least $\mu_0\left[F_1\omega^H+\left(1-F_1\right)\left(p\omega^H-c_d\right)\right]$. On the other hand, by Proposition \ref{Prop_BM_2sided}, which relies on Assumption \ref{assum_2sided}-(\ref{assum_convex}), (\ref{assum_sb1_F1}), and (\ref{assum_sb1_1/2}), the burning-money setting with commitment yields either
    \begin{eqnarray*}
        \mu_0\left[F_1\omega^H+\left(1-F_1\right)\left(p\omega^H-c_d\right)-s^b_1\right]~~\text{or}~~\mu_0\left[F_1\omega^H+\left(1-F_1\right)\left(p\omega^H-c_d\right)\right]-(1-\mu_0)s^b_1.
    \end{eqnarray*}
    It is straightforward to see that both of them are strictly smaller than the ex ante payoff in the audience-cost case for any $\mu_0\in(0,1)$.

    Proposition \ref{Prop_BM_2sided} shows that commitment power improves the defender's ex ante payoff. Thus, the payoff from burning money without commitment must be smaller than that from audience cost.
\end{proof}

\begin{proof}[Proof of Proposition \ref{AC2_Limitation}]
    Consider the following cases.
    
    \textbf{Step 1} (Separating). By the assumption that $\overline{s}^t<s^*$, there does not exist a separating equilibrium. This is because the low type would profitably deviate from $0$ to a positive $s$ that the high type sends, and he would receive $F_1\omega^L-(1-F_1)s>0$.
    
    We first show that any equilibrium in this case is either a pooling equilibrium or a semi-separating equilibrium in which the high type sends $s=s^\prime>0$ with probability one and the low type mixes $s=s^\prime$ and $s=0$. First, if the high type mixes, there cannot be a signal that comes only from the high type. This is because, if a signal $s$ comes only from the high type, the posterior becomes $\mu_s=1$. Then, because the high type never pays audience costs in equilibrium, a signal affects the high type's payoff only through posterior beliefs. Thus, the high type would strictly prefer $s$ such that $\mu_s=1$ and send it with probability one. Second, when the low type mixes, the signal sent only by the low type must be $s=0$ because such a signal leads to a posterior $\mu_s=0$ and any $s>0$ would result in a strictly negative payoff. Third, suppose that two different signals are sent by both types with positive probability, denoted by $s^\prime<s^{\prime\prime}$. Because the high type never pays audience costs on any equilibrium path, any difference in an equilibrium payoff must come from different posteriors. However, the high type's indifference condition imposes $\mu_{s=s^\prime}=\mu_{s=s^{\prime\prime}}$. Then, because the posterior is the same, the low type would strictly prefer $s^\prime$ to $s^{\prime\prime}$ by $s^\prime<s^{\prime\prime}$. Consequently, any equilibrium must be a pooling equilibrium or a semi-separating equilibrium in which the high type sends some $s=s^\prime>0$ and the low type mixes $s=s^\prime$ and $s=0$.

    \textbf{Step 2} (Semi-separating). Consider a semi-separating equilibrium where the high-valuation defender sends $s=s^\prime>0$ and the low type sends $s=s^\prime$ with probability $q\in(0,1)$ and $s=0$ with probability $1-q$. The low type's indifference condition yields
    \begin{eqnarray*}
        F\left(v^*_c(\mu_{s=s^\prime})\right)\omega^L-\left(1-F\left(v^*_c(\mu_{s=s^\prime})\right)\right)s^\prime &=& 0\\
        F\left(v^*_c(\mu_{s=s^\prime})\right) &=& \frac{s^\prime}{\omega^L+s^\prime}.
    \end{eqnarray*}

    We first show that the D1 criterion requires $s^\prime=\overline{s}^t$. Suppose $s^\prime<\overline{s}^t$, and consider $s=\overline{s}^t$ as an off-path signal. Because the low type's equilibrium payoff is zero, he weakly benefits from deviating to $s=\overline{s}^t$ if $F\left(v_c^*(\mu_{s=\overline{s}^t})\right)\ge\frac{\overline{s}^t}{\omega^L+\overline{s}^t}$. On the other hand, because the high type never pays audience costs in any equilibrium, he strictly benefits from deviating to $s=\overline{s}^t$ if $F\left(v_c^*(\mu_{s=\overline{s}^t})\right)>F\left(v_c^*(\mu_{s=s^\prime})\right)=\frac{s^\prime}{\omega^L+s^\prime}$. By $\frac{s^\prime}{\omega^L+s^\prime}<\frac{\overline{s}^t}{\omega^L+\overline{s}^t}$, the high type strictly benefits from this deviation whenever the low type weakly benefits from it. Then, the D1 criterion assigns $\mu_{s=\overline{s}^t}=1$. By $\overline{s}^t<\frac{F_1}{1-F_1}\,\omega^L$, the high type profitably deviates to $s=\overline{s}^t$ because we have $\frac{s^\prime}{\omega^L+s^\prime}<\frac{\overline{s}^t}{\omega^L+\overline{s}^t}<F_1$. Thus, any semi-separating equilibrium satisfying the D1 criterion must have $s^\prime=\overline{s}^t$.
    
    Next, by Bayes' rule, we have $\mu_{s=\overline{s}^t}=\frac{\mu_0}{\mu_0+( 1-\mu_0)q}\in(\mu_0,1)$ for $q\in(0,1)$. Thus, for this semi-separating equilibrium to exist, it must be that $F\left(v^*_c(\mu_{s=\overline{s}^t})\right)=\frac{\overline{s}^t}{\omega^L+\overline{s}^t}\in(F_{\mu_0},F_1)$. These conditions are equivalent to $\frac{F_{\mu_0}}{1-F_{\mu_0}}\,\omega^L<\overline{s}^t<\frac{F_1}{1-F_1}\omega^L$. By $f(x)>0$ and $\mu_0<1$, we have $F_{\mu_0}<F_1$, and thus $1-F_{\mu_0}$ is strictly positive. Thus, the threshold $\frac{F_{\mu_0}}{1-F_{\mu_0}}\,\omega^L$ is well defined; by the assumption that $\overline{s}^t<s^*$, the second inequality holds.
    
    Finally, we verify that the semi-separating equilibrium with $s^\prime=\overline{s}^t$ satisfies the D1 criterion. Consider any off-path signal $s<\overline{s}^t$. The high type weakly benefits from a smaller $s$ if $F\left(v^*_c(\mu_{s<\overline{s}^t})\right)\ge\frac{\overline{s}^t}{\omega^L+\overline{s}^t}$. The low type strictly benefits from deviating to a smaller $s$ if $F\left(v^*_c(\mu_{s<\overline{s}^t})\right)>\frac{s}{\omega^L+s}$. By $\frac{s}{\omega^L+s}<\frac{\overline{s}^t}{\omega^L+\overline{s}^t}$ for $s<\overline{s}^t$, the low type has a strictly profitable deviation to a small signal whenever the high type has a weakly profitable deviation. Thus, the D1 criterion imposes $\mu_s=0$ for $s<\overline{s}^t$. The high type would receive $p\omega^H-c_d$ from this deviation, which is strictly smaller than his equilibrium payoff, $\frac{\overline{s}^t}{\omega^L+\overline{s}^t}\,\omega^H+\left(1-\frac{\overline{s}^t}{\omega^L+\overline{s}^t}\right)\left(p\omega^H-c_d\right)$. Under $\mu_{s<\overline{s}^t}=0$, the low type would receive $-s\le0$. Thus, neither type profitably deviates, and the semi-separating equilibrium satisfies the D1 criterion.

    \textbf{Step 3} (Pooling). Consider a pooling equilibrium in which both types send $s=s^\prime\in S^t$. On the equilibrium path, the posterior is $\mu_{s=s^\prime}=\mu_0$. We first show that the D1 criterion requires $s^\prime=\overline{s}^t$. Suppose to the contrary that $s^\prime<\overline{s}^t$, and consider the off-path signal $s=\overline{s}^t$. The low type weakly benefits from deviating to $\overline{s}^t$ if
    \begin{eqnarray*}
        F\left(v^*_c(\mu_{s=\overline{s}^t})\right)\omega^L-\left(1-F\left(v^*_c(\mu_{s=\overline{s}^t})\right)\right)\overline{s}^t &\ge& F_{\mu_0}\omega^L-(1-F_{\mu_0})s^\prime\\
        F\left(v^*_c(\mu_{s=\overline{s}^t})\right) &\ge& \frac{F_{\mu_0}\omega^L-(1-F_{\mu_0})s^\prime+\overline{s}^t}{\omega^L+\overline{s}^t}\\
        &=&F_{\mu_0}+\frac{(1-F_{\mu_0})\left(\overline{s}^t-s^\prime\right)}{\omega^L+\overline{s}^t}>F_{\mu_0}.
    \end{eqnarray*}
    The high type strictly benefits from deviating to $\overline{s}^t$ if $F\left(v^*_c(\mu_{s=\overline{s}^t})\right)>F_{\mu_0}$. Thus, the high type has a strictly profitable deviation to $\overline{s}^t$ whenever the low type weakly benefits from the deviation. The D1 criterion imposes $\mu_{s=\overline{s}^t}=1$. Then, the high type strictly benefits by deviating to $s=\overline{s}^t$. As a result, any pooling equilibrium satisfying the D1 criterion must have $s^\prime=\overline{s}^t$.

    Because the low type can guarantee a payoff of zero by sending $s=0$, a pooling equilibrium in which both types send $s=\overline{s}^t$ exists when $F_{\mu_0}\omega^L-(1-F_{\mu_0})\overline{s}^t\ge0$, which is equivalent to $\overline{s}^t\le\frac{F_{\mu_0}}{1-F_{\mu_0}}\,\omega^L$.
    
    Lastly, we show that the pooling equilibrium with $s^\prime=\overline{s}^t$ satisfies the D1 criterion. Consider an off-path signal $s<\overline{s}^t$. The high type weakly benefits from deviating if $F\left(v^*_c(\mu_s)\right)\ge F_{\mu_0}$. The low type strictly benefits from deviating to $s<\overline{s}^t$ if
    \begin{eqnarray*}
        F\left(v^*_c(\mu_s)\right)>F_{\mu_0}-\frac{(1-F_{\mu_0})\left(\overline{s}^t-s\right)}{\omega^L+s}.
    \end{eqnarray*}
    Because $F_{\mu_0}-\frac{(1-F_{\mu_0})\left(\overline{s}^t-s\right)}{\omega^L+s}<F_{\mu_0}$, the low type has a strictly profitable deviation to $s<\overline{s}^t$ whenever the high type has a weakly profitable deviation to such $s$. Thus, the D1 criterion assigns $\mu_s=0$. Under this belief, the high type would receive $p\omega^H-c_d$, which is strictly smaller than his equilibrium payoff, $F_{\mu_0}\omega^H+(1-F_{\mu_0})\left(p\omega^H-c_d\right)$. The low type would receive $-s\le0$, which never exceeds his equilibrium payoff, $F_{\mu_0}\omega^L-(1-F_{\mu_0})\overline{s}^t\ge0$. Thus, this pooling equilibrium satisfies the D1 criterion.

    \textbf{Step 4} (defender-optimal equilibrium). By Step 1, we know that the only candidates for defender-optimal equilibria are a semi-separating equilibrium where the low type strictly mixes between $s=\overline{s}^t$ and $s=0$, and a pooling equilibrium where both types of the defender send $s=\overline{s}^t$. Step 2 shows that the semi-separating equilibrium exists if and only if $\overline{s}^t>\frac{F_{\mu_0}}{1-F_{\mu_0}}\,\omega^L$. We also know from Step 3 that the pooling equilibrium exists if and only if $\overline{s}^t\le\frac{F_{\mu_0}}{1-F_{\mu_0}}\,\omega^L$. Hence, these equilibria are defender-optimal whenever they exist.
\end{proof}

\begin{proof}[Proof of Proposition \ref{Comparison_BMwoComm_ACwoComm2_Limitation}]
    First, by Proposition \ref{AC2_Limitation}, the ex ante payoffs from the semi-separating and pooling equilibria are, respectively,
    \begin{eqnarray}
        \mu_0\left[\frac{\overline{s}^t}{\omega^L+\overline{s}^t}\omega^H+\left(1-\frac{\overline{s}^t}{\omega^L+\overline{s}^t}\right)\left(p\omega^H-c_d\right)\right]+(1-\mu_0)\times0~\text{and} \label{eqn:smallAC-semi-payoff}\\
        \mu_0\left[F_{\mu_0}\omega^H+(1-F_{\mu_0})\left(p\omega^H-c_d\right)\right]+(1-\mu_0)\left[F_{\mu_0}\omega^L-(1-F_{\mu_0})\overline{s}^t\right]. \label{eqn:smallAC-pool-payoff}
    \end{eqnarray}
    These values become equal at $\mu_0$ such that $F_{\mu_0}=\frac{\overline{s}^t}{\omega^L+\overline{s}^t}$. The former value is greater if and only if $F_{\mu_0}<\frac{\overline{s}^t}{\omega^L+\overline{s}^t}$, which is equivalent to the condition under which the semi-separating equilibrium in the audience-cost setting with $\overline{s}^t<s^*$ exists, i.e., $\overline{s}^t>\frac{F_{\mu_0}}{1-F_{\mu_0}}\,\omega^L$.

    We now compare the above payoffs with the payoff from the separating equilibrium in the burning-money setting without commitment. By Proposition \ref{benchmark_BM2}, observe that the ex ante payoff in the burning-money setting without commitment is
    \begin{eqnarray}\label{eqn:BMwoC-payoff}
        \mu_0\left[F_1\omega^H+(1-F_1)\left(p\omega^H-c_d\right)-F_1\omega^L\right]+(1-\mu_0)\times0.
    \end{eqnarray}
    Rearranging the payoffs in expressions \eqref{eqn:BMwoC-payoff} and \eqref{eqn:smallAC-semi-payoff} shows that the payoff from the semi-separating equilibrium in the small-audience-cost setting is strictly greater than the payoff from the separating equilibrium in the burning-money setting without commitment if and only if $\overline{s}^t>\frac{F_1D}{\delta-F_1D}\,\omega^L$. Because the payoff from the pooling equilibrium with $s=\overline{s}^t$ in the small-audience-cost setting is equal to or strictly greater than the payoff from the semi-separating equilibrium whenever the former exists, the audience-cost setting yields a strictly higher defender ex ante payoff than the burning-money setting without commitment if $\overline{s}^t>\frac{F_1D}{\delta-F_1D}\,\omega^L$.

    Next, suppose $\overline{s}^t<\frac{F_1D}{\delta-F_1D}\,\omega^L$. Because the separating equilibrium yields a strictly greater ex ante payoff than the semi-separating equilibrium for any $\mu_0>0$, we compare the payoffs from the separating equilibrium in the burning-money setting without commitment and the pooling equilibrium in the small-audience-cost setting. Observe that, by expressions \eqref{payoff_uncertain_BM} and \eqref{payoff_uncertain_AC}, the interim payoffs of the two settings are equal when $s=0$, i.e., $\widehat v^b(\mu,0)=\widehat v^t(\mu,0)$. Thus, by expression \eqref{eqn:smallAC-pool-payoff}, we can rewrite the ex ante payoff in the pooling equilibrium as $\widehat v^t\left(\mu_0,\overline{s}^t\right)=\widehat v^b(\mu_0,0)-(1-\mu_0)(1-F_{\mu_0})\overline{s}^t$. By expression \eqref{eqn:BMwoC-payoff}, the ex ante payoff in the burning-money setting without commitment can be rewritten as $\mu_0\left(\widehat v^b(1,0)-F_1\omega^L\right)$. Then, by rearranging these payoffs, we have
    \begin{eqnarray*}
        \frac{\widehat v^b(\mu_0,0)}{\mu_0}=\widehat v^b(1,0)-F_1\omega^L+\frac{1-\mu_0}{\mu_0}(1-F_{\mu_0})\overline{s}^t
    \end{eqnarray*}
    when the payoffs are equal. By Assumption \ref{assum_2sided}-(\ref{assum_convex}), the left-hand side is strictly increasing in $\mu_0$. On the other hand, the right-hand side is strictly decreasing in $\mu_0$ because $F_{\mu_0}$ is increasing and $\frac{1-\mu_0}{\mu_0}$ is decreasing in $\mu_0$. Because $\widehat v^t\left(0,\overline{s}^t\right)=-\overline{s}^t<0$ and $\widehat v^t\left(1,\overline{s}^t\right)=\widehat v^b(1,0)>1\times\left(\widehat v^b(1,0)-F_1\omega^L\right)$, there is a unique $\overline{\mu}_0\in(0,1)$ at which the two payoffs are equal and above which the audience-cost setting yields a strictly greater ex ante payoff.

    Finally, suppose $\overline{s}^t=\frac{F_1D}{\delta-F_1D}\,\omega^L$. In this case, observe that the ex ante payoff in the burning-money setting without commitment is equal to that of the semi-separating equilibrium in the audience-cost setting. Hence, the two settings yield the same ex ante payoff for $\mu_0\in\left[0,\mu_0^\prime\right]$ where $F_{\mu_0^\prime}=\frac{\overline{s}^t}{\omega^L+\overline{s}^t}$, and the audience-cost setting achieves a strictly higher payoff for $\mu_0>\mu_0^\prime$.
\end{proof}

\begin{proof}[Proof of Proposition \ref{Comparison_BMwComm_ACwoComm2_Limitation}]
    By Proposition \ref{AC2_Limitation}, the defender-optimal equilibrium in the audience-cost setting with $\overline{s}^t<s^*$ is either a semi-separating equilibrium or a pooling equilibrium in which both types of the defender send $s=\overline{s}^t$.

    First, suppose that the semi-separating equilibrium is defender-optimal in the small-audience-cost setting. By Proposition \ref{AC2_Limitation}, we know that the payoff from the semi-separating equilibrium is linear in $\mu_0$. In the burning-money setting with commitment, we also know that, for any pair of $\mu_0^\prime\le1/2$ and $\mu_0^{\prime\prime}>1/2$, the slope of $V^{jo}\left(\mu_0^{\prime\prime}\middle|\left(\widehat v^b_s\right)_{s\in S^b}\right)$ is greater than that of $V^{jo}\left(\mu_0^\prime\middle|\left(\widehat v^b_s\right)_{s\in S^b}\right)$. Hence, burning money with commitment is strictly more beneficial for any $\mu_0\in(0,1)$ whenever it is beneficial for $\mu\in(0,1/2]$. This is true when
    \begin{eqnarray*}
        && \overbrace{\mu_0\left[F_1\omega^H+\left(1-F_1\right)\left(p\omega^H-c_d\right)-s^b_1\right]+(1-\mu_0)\times0}^{V^{jo}\left(\mu_0\le\frac{1}{2}\middle|\left(\widehat v^b_s\right)_{s\in S^b}\right)}\\
        &>& \underbrace{\mu_0\left[\frac{\overline{s}^t}{\omega^L+\overline{s}^t}\omega^H+\left(1-\frac{\overline{s}^t}{\omega^L+\overline{s}^t}\right)\left(p\omega^H-c_d\right)\right]+(1-\mu_0)\times0}_{\text{Semi-separating in audience costs}}.
    \end{eqnarray*}
    This condition simplifies to
    \begin{eqnarray*}
        s^b_1<\left(F_1-\frac{\overline{s}^t}{\omega^L+\overline{s}^t}\right)\left((1-p)\omega^H+c_d\right),
    \end{eqnarray*}
    which is Assumption \ref{assum_2sided}-(\ref{assum_sb1_sbart}).

    Second, suppose the pooling equilibrium is defender-optimal in the small-audience-cost setting. By the proof of Proposition \ref{Comparison_BMwoComm_ACwoComm2_Limitation}, we know that $\widehat v^t\left(\mu_0,\overline{s}^t\right)=\widehat v^b(\mu_0,0)-(1-\mu_0)(1-F_{\mu_0})\overline{s}^t<\widehat v^b(\mu_0,0)$. Moreover, Proposition \ref{Prop_BM_2sided} shows that the ex ante payoff in the burning-money setting with commitment is strictly greater than $\widehat v^b(\mu_0,0)$ for any $\mu_0\in(0,1)$ under Assumption \ref{assum_2sided}-(\ref{assum_convex}) and (\ref{assum_sb1_1/2}). Therefore, the defender's ex ante payoff of the burning-money setting with commitment is strictly greater than that of the audience-cost setting with $\overline{s}^t<s^*$ for any $\mu_0\in(0,1)$.
\end{proof}

\begin{proof}[Proof of Proposition \ref{WarProb}]
    First, consider the burning-money setting.
    The results in Propositions \ref{benchmark_BM2} and \ref{Prop_BM_2sided} use Assumption \ref{assum_2sided}-(\ref{assum_convex}), (\ref{assum_sb1_F1}), (\ref{assum_sb1_1/2}), (\ref{assum_rich_F1}), and (\ref{assum_rich_s+Delta}).
    By Propositions \ref{benchmark_BM2} and \ref{Prop_BM_2sided}, the defender conveys his type in the defender-optimal equilibrium in the burning-money setting with and without commitment. Thus, war occurs when a challenger with $v_c>v^*_c(1)$ challenges and the high type fights. The ex ante probability of war in this case is $\mu_0\times(1-F_1)$.

    Second, consider the audience-cost setting with $\overline{s}^t\ge s^*$. When $-p\omega^L+c_d \ge \frac{F_1}{1-F_1}\, \omega^L$, by Proposition \ref{benchmark_AC2_BridgeBurning}, a separating equilibrium is defender-optimal. Thus, as in the burning-money setting, the ex ante probability of war is $\mu_0(1-F_1)$, which is equal to the probability of war in the burning-money setting.
    
    When $-p\omega^L+c_d \le \frac{F_1}{1-F_1}\, \omega^L$, the defender-optimal equilibrium becomes the one in which both types commit to fighting by sending $s\ge-p\omega^L+c_d$. Thus, war erupts when a challenger with $v_c>v^*_c(1)$ challenges. Then, the ex ante probability of war is $1-F_1$. Therefore, when $-p\omega^L+c_d < \frac{F_1}{1-F_1}\,\omega^L$, the probability of war in the large audience-cost setting is strictly greater than that of the burning-money setting for any $\mu_0<1$ and any $F_1<1$. The special case of $F_1=1$ belongs to the condition $-p\omega^L+c_d \le \frac{F_1}{1-F_1}\, \omega^L$ because $\frac{F_1}{1-F_1}\,\omega^L=\infty$. In this case, the burning-money and audience-cost settings yield the same ex ante probability of war by $\mu_0(1-F_1)=1-F_1=0$ when $F_1=1$.
    
    When $-p\omega^L+c_d=\frac{F_1}{1-F_1}\, \omega^L$, both of the above equilibria are optimal. Moreover, there also exist semi-separating equilibria in which the low type mixes $s\ge-p\omega^L+c_d$ and $s=0$. Case 3 of the proof of Proposition \ref{benchmark_AC2_BridgeBurning} shows that the low type may commit to fighting with any probability $q\in[0,1]$, whereas the high type commits to fighting for sure. Consequently, the ex ante probability of war is $\left(\mu_0+(1-\mu_0)q\right)(1-F_1)\in[\mu_0(1-F_1),1-F_1]$. Therefore, when $-p\omega^L+c_d=\frac{F_1}{1-F_1}\, \omega^L$, the audience-cost setting yields the same probability of war as the burning-money setting if $q=0$ and a strictly greater probability of war otherwise.

    Third, consider the audience-cost setting with $\overline{s}^t<s^*$. By Proposition \ref{AC2_Limitation}, when $\overline{s}^t\in\left(\frac{F_{\mu_0}}{1-F_{\mu_0}}\,\omega^L,s^*\right)$, the defender-optimal equilibrium is a semi-separating equilibrium where the low type mixes between $s=0$ and $s=\overline{s}^t$. By $F\left(v^*_c(\mu_{s=\overline{s}^t})\right)=\frac{\overline{s}^t}{\omega^L+\overline{s}^t}$ as shown in Proposition \ref{AC2_Limitation}, the ex ante probability of war is $\mu_0\left(1-F\left(v^*_c(\mu_{s=\overline{s}^t})\right)\right)=\frac{\omega^L}{\omega^L+\overline{s}^t}\,\mu_0$. This is strictly greater than the probability of war in the burning-money setting because we have $1-F\left(v^*_c(\mu_{s=\overline{s}^t})\right)>1-F_1$ when $\overline{s}^t<\frac{F_1}{1-F_1}\,\omega^L$.
    
    By Proposition \ref{AC2_Limitation}, the defender-optimal equilibrium is a pooling equilibrium with $s=\overline{s}^t$ when $\overline{s}^t\le\frac{F_{\mu_0}}{1-F_{\mu_0}}\,\omega^L$. Thus, war occurs when the challenger with $v_c>v^*_c(\mu_0)$ challenges and then the high type fights. Hence, the probability of war is $\mu_0(1-F_{\mu_0})$, which is greater than the probability of war in the burning-money setting, $\mu_0(1-F_1)$, by $F_{\mu_0}<F_1$. In sum, when $\overline{s}^t<s^*$, the ex ante probability of war is strictly greater in the audience-cost setting than in the burning-money setting for any $\mu_0\in(0,1)$.
\end{proof}

\section{Audience Costs with Commitment}\label{app:C}

For completeness, here we characterize the audience-cost setting with commitment power.

\begin{prop}\label{Prop_AC_2sided}
    Suppose Assumption \ref{assum_2sided}-(\ref{assum_convex}) holds.
        \begin{itemize}
            \item When $-p\omega^L+c_d \ge \frac{F_1}{1-F_1}\, \omega^L$ or $\overline{s}^t<-p\omega^L+c_d$, the defender sends some $s>0$ when $\omega=\omega^H$ and sends $s=0$ when $\omega=\omega^L$. The challenger challenges when $s=0$ or $v_c>\frac{c_c}{1-p}$, upon which the defender fights if and only if $\omega=\omega^H$.
            \item When $-p\omega^L+c_d \le \frac{F_1}{1-F_1}\, \omega^L$ and $\overline{s}^t\ge-p\omega^L+c_d$, both types of the defender send some $s\ge c_d-p\omega^L$. The challenger challenges when $s<c_d-p\omega^L$ or $v_c>\frac{c_c}{1-p}$. Both types of the defender fight when challenged on the equilibrium path.
        \end{itemize}
\end{prop}

\begin{proof}[Proof of Proposition \ref{Prop_AC_2sided}]
    We first specify the signal structure and then show that it achieves the defender's highest possible ex ante payoff.
    
    \textbf{Step 1} (Signal structure). First, suppose $-p\omega^L+c_d \ge \frac{F_1}{1-F_1}\, \omega^L$ or $\overline{s}^t<-p\omega^L+c_d$. Consider a separating signal structure in which the low type sends $s=0$ and the high type sends a fixed feasible signal $s>0$. The corresponding posteriors are $\mu_{s=0}=0$ and $\mu_{s>0}=1$. Bayes plausibility requires
    \begin{eqnarray*}
        && \tau(\mu_{s=0})\mu_{s=0}+(1-\tau(\mu_{s=0}))\mu_{s>0}=\mu_0\\
        &\Longrightarrow& \tau(\mu_{s=0})=1-\mu_0~\text{and}~\tau(\mu_{s>0})=\mu_0.
    \end{eqnarray*}
    The defender's signal structure follows from the above values:
    \begin{eqnarray*}
        \pi(0|\omega^H)=0,~\pi(s>0|\omega^H)=1,~\pi(0|\omega^L)=1,~\text{and}~\pi(s>0|\omega^L)=0.
    \end{eqnarray*}
    As a result, the challenger challenges when $s=0$ or $v_c>\frac{c_c}{1-p}$, upon which the high-valuation defender fights.

    Next, suppose $-p\omega^L+c_d \le \frac{F_1}{1-F_1}\, \omega^L$ and $\overline{s}^t\ge-p\omega^L+c_d$. As in the case without commitment (Proposition \ref{benchmark_AC2_BridgeBurning}), the defender is better off committing to fighting ($s\ge-p\omega^L+c_d$) rather than separating. That is, when
    \begin{eqnarray*}
        &&\overbrace{F_1\left(\mu_0\omega^H+(1-\mu_0)\omega^L\right)+(1-F_1)\left(p\left(\mu_0\omega^H+(1-\mu_0)\omega^L\right)-c_d\right)}^{\text{Committing to fighting}}\\
        &\ge& \underbrace{\mu_0\left(F_1\omega^H+(1-F_1)W^H\right),}_{\text{Committing to separating}}
    \end{eqnarray*}
    committing to fighting is weakly better than the separating signal structure. Rearranging, we obtain the condition $-p\omega^L+c_d \le \frac{F_1}{1-F_1}\, \omega^L$.

    \textbf{Step 2} (Payoff optimality). The above signal structures yield the highest possible ex ante payoff in each case. First, suppose $-p\omega^L+c_d \ge \frac{F_1}{1-F_1}\, \omega^L$ or $\overline{s}^t<-p\omega^L+c_d$. The ex ante payoff in the case with commitment is $\mu_0\widehat v^t(1,0)=\mu_0\widehat v^b(1,0)$. By $\widehat v^t(\mu,s)\le\widehat v^b(\mu,0)$ for $s<-p\omega^L+c_d$, Assumption \ref{assum_2sided}-(\ref{assum_convex}) guarantees $\widehat v^t(\mu,s)\le\widehat v^b(\mu,0)\le\mu\widehat v^b(1,0)$. When some $s\ge-p\omega^L+c_d$ is feasible, note that the condition $-p\omega^L+c_d\ge\frac{F_1}{1-F_1}\,\omega^L$ is equivalent to $F_1\omega^L+(1-F_1)W^L\le0$. Thus, for any $s\ge-p\omega^L+c_d$, we have
    \begin{eqnarray*}
        \widehat v^t\left(\mu,s\ge-p\omega^L+c_d\right)=\mu\widehat v^t(1,0)+(1-\mu)\left(F_1\omega^L+(1-F_1)W^L\right)\le\mu\widehat v^t(1,0).
    \end{eqnarray*}
    In sum, by Bayes plausibility, we obtain the ex ante payoffs
    \begin{eqnarray*}
        \sum_{s\in S^t}\tau(\mu_s)\widehat v^t(\mu_s,s)\le\sum_{s\in S^t}\tau(\mu_s)\mu_s\widehat v^t(1,0)=\mu_0\widehat v^t(1,0).
    \end{eqnarray*}
    This upper bound in the ex ante payoff is attained when the high type sends some $s>0$ and the low type sends $s=0$.
    
    Next, suppose $-p\omega^L+c_d \le \frac{F_1}{1-F_1}\, \omega^L$ and $\overline{s}^t\ge-p\omega^L+c_d$. Because $F_1\omega^L+(1-F_1)W^L\ge0$ in this case, it is straightforward to see
    \begin{eqnarray*}
        \widehat v^t(\mu,s\ge-p\omega^L+c_d)=\mu\widehat v^t(1,0)+(1-\mu)\left(F_1\omega^L+(1-F_1)W^L\right)
    \end{eqnarray*}
    for any $s\ge-p\omega^L+c_d$ and $\mu$ because the low type also fights for sure if challenged and the Challenger knows that both types will fight. Moreover, for $s<-p\omega^L+c_d$, we have $\widehat v^t\left(\mu,s<-p\omega^L+c_d\right)\le\mu\widehat v^t(1,0)\le\mu\widehat v^t(1,0)+(1-\mu)\left(F_1\omega^L+(1-F_1)W^L\right)$
    In sum, by Bayes plausibility, we have the ex ante payoffs
    \begin{eqnarray*}
        \sum_{s\in S^t}\tau(\mu_s)\widehat v^t(\mu_s,s) &\le& \sum_{s\in S^t}\tau(\mu_s)\mu_s\widehat v^t(1,0)+\sum_{s\in S^t}\tau(\mu_s)(1-\mu_s)\left(F_1\omega^L+(1-F_1)W^L\right)\\
        &=& \mu_0\widehat v^t(1,0)+(1-\mu_0)\left(F_1\omega^L+(1-F_1)W^L\right),
    \end{eqnarray*}
    which establishes the payoff optimality. This upper bound is attained when both types send some $s\ge-p\omega^L+c_d$ and fight if challenged.
Therefore, committing to fighting is defender-optimal.
    When $-p\omega^L+c_d = \frac{F_1}{1-F_1}\, \omega^L$ and $\overline{s}^t\ge-p\omega^L+c_d$, both signaling structures are optimal.
\end{proof}

\end{document}